\documentclass[11pt,final]{article}
\usepackage{fullpage}
\usepackage{nicefrac}
\usepackage{xfrac}
\usepackage{amsmath,amssymb,amsthm}
\usepackage{booktabs} 
\usepackage{xspace}
\usepackage{accents} 
\usepackage{pgf} 
\usepackage{xparse} 
\usepackage{bbm}
\usepackage{enumerate}
\usepackage{algorithm}
\usepackage[noend]{algorithmic}
\usepackage{threeparttable}
\usepackage{array}
\usepackage{booktabs}
\usepackage{multicol}
\usepackage{tabularx}
\usepackage{framed}
\usepackage{thmtools, thm-restate}
 \usepackage{mathtools}
 \usepackage[square,longnamesfirst]{natbib}
\usepackage{verbatim}
\usepackage{tabu, caption}
\usepackage{fancybox,framed}
\usepackage{mathtools}

\usepackage[colorlinks=true]{hyperref}
\hypersetup{
	linkcolor=[rgb]{0.4,0,0.6},
	citecolor=[rgb]{0, 0.4, 0},
	urlcolor=[rgb]{0.6, 0, 0}
}

\newtheorem{lemma}{Lemma}

\newtheorem{corollary}{Corollary}
\newtheorem{theorem}{Theorem}

\theoremstyle{definition}
\newtheorem{definition}{Definition}

\newtheorem{thm}{Theorem}
\newtheorem{observation}[thm]{Observation}

\newcommand{\expect}[2]{\mathop{\mathbb{E}}_{#1} \left[ #2 \right]}

\usepackage{todonotes}

\usepackage{accsupp}

\newcommand{\MyFrame}[1]{\noindent \framebox[\textwidth]{ \begin{minipage}{0.97\textwidth} #1 \end{minipage}}}%

\begin{document}
\title{Prompt Scheduling for Selfish Agents\footnote{The work of A. Eden, M. Feldman and T. Taub was partially supported by the European Research Council under the European Union's Seventh Framework Programme (FP7/2007-2013) / ERC grant agreement number 337122, and by the Israel Science Foundation (grant number 317/17). The work of A. Eden, A. Fiat and T. Taub was partially supported by ISF 1841/14.}}
\author{Alon Eden%
\thanks{%
    {Tel Aviv University (\url{alonarden@gmail.com})}}
\and Michal Feldman%
\thanks{%
    {Tel Aviv University and Microsoft Research Israel (\url{michal.feldman@cs.tau.ac.il})}}
\and Amos Fiat%
\thanks{%
    {Tel Aviv University (\url{fiat@tau.ac.il})}}    
\and Tzahi Taub%
\thanks{%
    {Tel Aviv University (\url{tzahita@gmail.com})}}    
}

\maketitle

\begin{abstract}
We give a prompt online mechanism for minimizing the sum of [weighted] completion times. This is the first  prompt online algorithm for the problem.
When such jobs are strategic agents, delaying scheduling decisions makes little sense. Moreover, the mechanism has a particularly simple
form of an anonymous menu of options.  

\end{abstract}

\section{Introduction}

The setting herein includes [multiple] service queues and selfish agents that arrive online over time and can be processed
on one of $m$ machines.
Agents may have some (private) {\sl processing time} $p$ and/or some private {\sl weight} $w$.

The goal is to improve service as much as possible. Minimizing the sum of [weighted] completion times is one measure of
how good (or bad) service really is.

This problem has long been studied, as a pure optimization problem, without strategic considerations \cite{Graham79}.
Given a collection of jobs, lengths, and weights, the shortest weighted processing time order \cite{Smith56}, also known as Smith's rule, produces
a minimal sum of weighted completion times with a non-preemptive schedule on a single machine.

Schedules can be preemptive (where jobs may be stopped and restarted over time) or non-preemptive (where
a job, once execution starts, cannot be stopped until the job is done).

To the best of our knowledge, all online algorithms for this problem have the following
property: when a job arrives, there are no guarantees as to when it will finish. If preemption is
allowed, even if the job starts, there is no guarantee that it will not be preempted, or for how long.
If preemption is disallowed, the online algorithm keeps the job "hanging about" for some unknown length of time,
until the algorithm finally decides that it is time to start it.

Essentially, this means that when one requests service, the answer is ``OK --- just hang around and you will get service at some unknown
future date". It is in fact impossible to achieve any bounded ratio for the sum of [weighted] completion times if one has to start processing
the job as soon as possible. Some delay is inevitable. However, the issue we address is ``does the job know when it will be served?".
All of these issues are fundamental when considering that every such ``job" is a strategic agent.
It is not only that one avoids uncertainty, knowing the future schedule allows one to make appropriate plans for the interim.

In this paper we present {\sl prompt} online algorithms that immediately determine as to when an incoming job will be processed
(without preemption).
The competitive ratio is the best possible, amongst all prompt online algorithms, even if randomization is allowed
(the algorithm is in fact deterministic). The competitive
ratio compares the sum of completion times of the online algorithm with the [harder to achieve] sum of completion times of an optimal
preemptive schedule.
Moreover, viewed in the context
of strategic agents, these scheduling algorithms are not only DSIC but of a particularly simple form.

Upon arrival, agents are presented with a menu of possible options, where a menu entry is of the form $([b,e],q,\pi)$. This means that
the period from $b$ to $e$ is available on machine $q$ and will cost the agent $\pi$. These menus are anonymous and
do not depend on the agent that arrives. The agent then chooses one of the options.

Rational agents will never choose an interval that
is shorter than the processing time. (If so the agent cost is $\infty$).
It is not hard to show that there is no advantage for an agent to delay her arrival.

The cost to the agent is the sum of two components: (a) The time spent waiting, weighted by the agents' [private] weight.
{\sl I.e.}, highly impatient agents will have high weight, less impatient agents will have lower weight.
(b) The price, $\pi$, associated with an option on the menu. Agents seek to minimize their cost.

Consider the case of a single queue, a selfish agent will simply join the queue immediately upon arrival, there is no reason to delay. Thus,
jobs will be processed in first-in-first-out (FIFO) order.
However, this may be quite bad in terms of the sum of completion times. Imagine a job with processing time $L$, arriving at time zero, followed by
$\sqrt{L}$ jobs of length $1$, all of which arrive immediately after the first.
As the first job will only be done at time $L$, the sum of completion times for these $1+\sqrt{L}$ jobs is about $L^{3/2}$.
Contrawise, if the $\sqrt{L}$ length one jobs were processed before the length $L$ job, the sum of completion times would be about $2L$.
Obviously it seems a good idea to delay longer jobs and expedite shorter jobs.

Similarly, consider a first batch of  $L$ jobs, each of length $1$ and weight $1$, immediately followed by a single job of length $1$ and weight
$W$. For FIFO processing, the weighted sum of completion times is $L^2/2$ (for the weight 1 jobs) plus $(L+1)\cdot W$ (for the job of weight $W$).
Optimally, the weight $W$ job should be processed first, followed by the length 1 jobs. The weighted sum of completion times is then
about $W+L^2/2$. For any constant $L$ and sufficiently large $W$, the ratio between the two sums approaches $L+1$.

The main question addressed in this paper is how to produce such dynamic menus so as to incentivize selfish agents
towards behavior that achieves some desirable social goal, specifically, minimizing the sum of completion times.
The dynamic menu is produced based on the past decisions of the previous agents and the current time\footnote{For clarity we describe
the menu as though it was infinite. In fact, one can think of the process as though the menu is
presented entry by entry. The selfish job will provably
choose an option early on.}.

We measure the quality of the solution achieved by the competitive ratio, the ratio between the sum of completion times of the selfish agents, when
presented with the dynamic menus, divided by the minimal sum of completion times, when the future arrivals and their private values are known.
In fact, the comparison is with the optimal preemptive schedule (which could definitely be better than the optimal non-preemptive schedule).

We consider several scenarios:\begin{enumerate}
\item All agents have weight 1 and arbitrary processing times, nothing known apriori on the processing times. This models cases where
all agents are equally impatient but have different processing requirements.
The underlying idea here is to offer menu options that delay longer jobs so that they do not overly
delay many shorter jobs that arrive later.
\item All agents have processing time 1 and arbitrary weight, nothing known apriori on the weights.
The underlying idea here is to set prices so as to delay jobs of small weight and thus to allow later jobs of large weight to finish early.
\item Jobs with arbitrary processing times and weights bounded by a known bound $B_{\max}$.
This means that we have to delay long jobs and simultaneously
have to leave available time slots for jobs with large weights.
\end{enumerate}

The competitive ratios for the different scenarios appear in Table \ref{tab:cr}. We remark that the lower bounds
hold even if one assumes that the machines used are arbitrarily faster than the machines used by the optimal schedule that minimizes the sum of weighted
completion times.

\begin{table}
{\footnotesize \centering
\begin{tabular}{| l | l | l | l | l |}
\hline
 \begin{tabular}{c}Processing \\
 Time \end{tabular} & \begin{tabular}{c}Job \\
 Weight \end{tabular} & \begin{tabular}{c}Menu \\
 entries \end{tabular}& \begin{tabular}{c}Upper \\
 Bound\\(Deterministic) \end{tabular}& \begin{tabular}{l}Lower \\
 Bound\\(Randomized) \end{tabular}  \\ \hline
 $p_j \in \mathbb{Z}^{+}$  & $w_j=1$ & \begin{tabular}{l} intervals\\(various lengths) \\no prices \end{tabular}
 & $O(\log P_{\max})$ & $\Omega(\log P_{\max})$  \\ \hline
 $p_j=1$  & $w_j\in \mathbb{Z}^{+}$ & \begin{tabular}{l} unit length\\intervals\\with prices \end{tabular}
 & $O(\log W_{\max} (\log \log W_{\max} +\log n))$ & $\Omega(\log W_{\max})$ \\ \hline
 $p_j \in \mathbb{Z}^{+}$  & $w_j\in \mathbb{Z}^{+}$
 & \begin{tabular}{l} intervals\\(various lengths) \\with prices \end{tabular}
 & $O\left((\log n + \log P_{\max}) \cdot \log B_{\max}\right)$ & $\Omega( \max( \log B_{\max}, \log P_{\max}))$ \\ \hline
 \end{tabular}}
\caption{Competitive Ratios of our Dynamic Menus, and associated lower bounds.
	$P_{\max}$ is the longest job processing time in the input sequence, it is not known apriori. $W_{\max}$ is the maximal job weight in the sequence, it is not known apriori. $B_{\max}$ is an apriori upper bound on $W_{\max}$. 
}
\label{tab:cr}
\end{table}

\subsection{Related Work}

For one machine, weighted jobs, available at time zero, ordering the jobs in order of weight/processing time minimizes the sum
of competition times \cite{Smith56}. For one machine, unweighted jobs with release times, a preemptive schedule that always processes the
job with the minimal remaining processing time minimizes the sum of weighted completion times \cite{Schrage66,Schrage68}. As an offline problem,
where jobs cannot be executed prior to some earliest time, finding an optimal non-preemptive schedule
is computationally hard \cite{Hall97}.

For parallel machines, where jobs arrive over time, a preemptive schedule that always processes the jobs with the highest priority:
 weight divided by remaining processing time, is a 2 approximation \cite{Megow04}, this algorithm is called {\sl weighted shortest remaining processing time} (WSRPT).
If all weights are one this preemptive algorithm is called {\sl shortest remaining processing time} (SRPT).
Other online and offline algorithms to minimize the
sum of completion times appear in \cite{Bruno74,Shmoys95,Hall97}.

\cite{Phillips1998} show how to convert a preemptive online algorithm into a non-preemptive online algorithm while increasing the
completion time of the job by no more than a constant factor. This transformation strongly depends on not determining immediately when the job
will be executed. This is in comparison to a prompt algorithm that determines when the job is executed immediately upon the job arrival.

When selfish agents are involved, it is valuable to keep things simple \cite{HR09}. Offering selfish agents an anonymous menu of options is an example of such a simple process. More complicated mechanisms require trust on the part of the agents.

Recently, \cite{FFR17} considered a similar question to ours, where a job with private processing time had to choose
between multiple FIFO queues, where the servers had different speeds.
Here, dynamic posted prices were associated with every queue, with the goal of [approximately] minimizing the
makespan, the length of time until the last job would finish.
Shortly thereafter, \cite{IMPS17} used dynamic pricing to minimize the maximal flow time.
Dynamic pricing schemes were considered for non-scheduling cost minimization problems in \cite{CohenEFJ15}.

A constant approximation mechanism for minimizing sum of completion times for selfish jobs was considered in \cite{GkatzelisMR17}, where the setting was an offline setting, the processing time was known in advance and the weight was private information. In an online setting, \cite{ImK16} show a constant approximation preemptive mechanism that gives an $O(1/\epsilon^2)$ approximation to the sum of flow times when using machines that are faster by a factor of $1+\epsilon$.

In this paper our goals are pricing schemes that affect agents as to behave in a manner
that [approximately] minimizes the sum of weighted completion times.

There is a vast body of work on machine scheduling problems, in offline and online settings, with strategic agents involved and not,
and in a host of models. It is impossible to do justice to this body of work but a very short list of additional relevant papers includes
\cite{G66,Lenstra1977,Graham79,Lenstra90,NR01,CKN04,ILMS09}.

\section{The Model}
We consider a job scheduling setting with $m$ machines and $n$ jobs that arrive in real time, where $p_j$, $w_j$, and $r_j$ are, respectively, the  processing time, weight, and release time of the $j$th job to arrive. It may be that $r_j=r_{j+1}$, {\sl i.e.}, more than one job arrive at the same time. However, job decisions are made sequentially in index order.

A valid input for this problem can be described as a sequence of jobs $$\sigma=(r_1,w_1,p_1), (r_2,w_2, p_2), \ldots,(r_n,w_n,p_n),$$ where the {\sl release time} $r_{i}\leq r_{i+1}$ for $i=1,\ldots,n-1$, the {\sl job weight} $w_i\geq 1$ for $i=1,\ldots,n$, and the job {\sl processing time} $p_i\geq 1$ for $i=1,\ldots,n$. We refer to the $j$th job in this sequence as job $j$. We use the terms {\sl size} and processing time interchangeably. Moreover, if $p_j<p_{j'}$ we may say job $j$ is smaller than job $j'$, etc. Let $\sigma[1..\ell]$ be the length $\ell$ prefix of $\sigma$. The total volume of a set of jobs $D$, denoted $vol(D)$ is the sum of processing times of the jobs in $D$, i.e., $vol(D)=\sum_{j\in D}p_j$.

Let $s_j\geq r_j$ be the time at which job $j$ starts processing (on some machine $1\leq q \leq m$). The completion time of job $j$ is $c_j = s_j+p_j$.

The objective considered in this paper is to minimize the sum of [weighted] completion times; {\sl i.e.}, we wish to minimize $\sum_{j=1}^{n} w_j\cdot c_j$.

For jobs $j$, $j'$, with $j<j'$ and with $r_j=r_{j'}$, job $j$ is assigned (or chooses) machine $q_j$ at time $s_j$ before job $j'$ is assigned machine $q_{j'}$ at $s_{j'}$. We say that $(q_j,s_j)$ and $(q_{j'},s_{j'})$ overlap, if $q_j=q_{j'}$ and ($s_j\leq s_{j'}<c_j=s_j+p_j$ or $s_{j'}\leq s_j <c_{j'}=s_{j'}+p_{j'}$).

A valid (non-preemptive) schedule for an input $\sigma$ is a sequence $$(m_1,s_1), (m_2,s_2), \ldots, (m_n,s_n)$$ where no overlaps occur. An online algorithm determines $(m_j,s_j)$ after seeing $\sigma[1\ldots j]$ and before seeing job $j+1$.

We consider online mechanisms where jobs are selfish agents, processing times and weights are private information, and job $j$ is presented with a menu of options upon arrival. Every option on the menu is of the form $(I,q,\pi)$ where (i) $I$ is a time interval $[b(I),e(I)]$, with integer endpoints, and where $b(I)\geq r_j$, (ii) $1 \leq q \leq m$ is some machine, and (iii) $\pi$ is the price for choosing this entry.  The menu of options presented to job $j$ is computed after jobs $1,\ldots,j-1$ have all made their choices and also depends on the release time of job $j$, $r_j$ (because one cannot process a job in the past). We assume {\sl no feedback} from jobs after they choose their menu options, i.e., if a job of size $p$ chooses an interval $I$ of length $|I|>p$, we do not know the interval is only partly used, and specifically, cannot offer the $|I|-p$ remaining to future jobs.

For job $j$ that chooses menu entry $([b(I),e(I)],q,\pi)$ we use the following notation (i) $I(j)$ for the interval chosen by job $j$, $[b(I),e(I)]$, (ii) $M(j)$ for the machine chosen by job $j$, $q$, and (iii) $\Pi(j)$ for the price of the entry chosen by $j$, $\pi$.

Although the menus described above are infinite, one can present the menu items sequentially.
With unit weight jobs, a job of processing time $p$ will make its choice within the first $\log p$ options presented.
With unit length jobs, a job of weight $w$ will make its choice within the first $\log w$ options presented.
With arbitrary lengths and arbitrary weights, a job of processing time $p$ and of weight $w$ will make its choice within the first $\log p \cdot \log w$ options presented.

The cost to job $j$ with weight $w_j$ and processing time $p_j$ for choosing the menu entry $([b,e],q,\pi)$ is $\infty$ if the time interval is too short: $e-b<p_j$. If $e-b\geq p_j$ then the cost to job $j$ is a cost of $w_j$ for every unit of time until job $j$ starts processing, plus the extra price from the menu. {\sl I.e.}, the cost to job $j$ with release time $r_j$, processing time $p_j$ and weight $w_j$, for choosing menu entry $([b,e],q,\pi)$, $e-b\geq p_j$, is $$(b+p_j)\cdot w_j + \pi.$$

For the specialized cases of weight one jobs or unit length jobs the general model above is somewhat simpler:

\subsection{Modeling weight one jobs with arbitrary Processing times}

If jobs have weight one, we give (optimal) menus that do not require pricing menu entries. Any entry on the menu is available for free. Therefore, we can simplify the menu structure as follows: 
 The job chooses a time interval and a machine from a menu with  entries of the form $([b,e],1\leq q \leq m)$ where the first entry is a time interval, and the second entry is a machine\footnote{Although the general setting allows pricing menu items, it turns out that for weight 1 jobs the optimal menu does not need to differentiate
 	entries by price.}.  The crux of the matter is coming up with the right menu.

Jobs choose from the menu one of the entries immediately upon arrival. As above,
we say that job $j$ chooses menu entry $(I(j),M(j))$ where $I(j)$ is an interval, and $1 \leq M(j) \leq m$.

For job $j$ with arrival time $r_j$, and processing time $p_j$ the cost associated with choosing the menu item
$([b,e],1\leq q \leq m)$ is $\infty$ if $p_j > e-b$ and $(b+p_j)$ otherwise. Jobs always seek to minimize their cost.

\subsection{Modeling unit length jobs of arbitrary weight}

Every job requires one unit of processing time on one of $m$ different processors. Every job $j$ is a selfish agent that has a private weight $w_j$, the cost to the job of one unit of delay.

The job chooses a machine and time slot from a menu with  entries of the form $([i,i+1],1\leq q \leq m, \pi)$ where the first entry is a time slot, the second entry is a machine, and the third entry is the price of this time slot on the machine.

Jobs choose from the menu one of the entries immediately upon arrival. Job $j$ is said to choose menu item $(I(j),M(j),\Pi(j))$ where $I(j)$ is a length one interval, $1 \leq M(j) \leq m$, and $\Pi(j)$ is the price to be paid for choosing this option.

For job $j$ with arrival time $r_j$, and weight $w_j$ the  cost associated with choosing the menu item
$([i,i+1],1\leq q \leq m, \pi)$ is  $w_j (i+1)+\pi$. Jobs always seek to minimize their cost.

\section{Dynamic Menu for Selfish Jobs with Heterogeneous Processing Times}\label{sec:dynamic}
In this section we introduce a dynamic menu based mechanism, for jobs of weight one and heterogeneous processing times, with competitive ratio $O(\log P_{\max})$, where $P_{\max}$ is the maximal job processing time among all jobs.

In Section~\ref{sec:natural_algs} we present a couple of natural algorithms that have competitive ratio of $\Omega(\sqrt{P_{\max}})$.
In Section~\ref{sec:integer-sequence} we provide integer sequences and corresponding interval sequences that serve as a building block for our dynamic menu mechanism, which is presented in Section~\ref{subsec:dynamic}.
Finally, in Section~\ref{sec:analysis} we provide the analysis showing that the dynamic menu gives a competitive ratio of $O(\log P_{\max})$.

\subsection{Warmup: non-working algorithms
}\label{sec:natural_algs}
We present two natural algorithms for prompt scheduling on a single machine, which result in poor competitive ratios. 
Assume $P_{\max}=2^d$, for some constant $d$, and $P_{\max}$ is known in advance. 
Assume also that all jobs have release time $0$ (but arrive sequentially). 
In this case, the optimal algorithm sorts jobs from short to long processing times, and schedules them based on this order.
In an attempt to mimic this optimal (offline) algorithm by an online algorithm --- in case where the input starts with a sequence of long jobs --- we would like to introduce delays, keeping some early intervals vacant for short jobs that might come in the future.

Consider an algorithm that sets a static interval sequence (i.e., a sequence that is set once and for all from the outset), and schedules each arriving job on the first interval on which it fits.

One natural algorithm sets (an infinite loop of) the following sequence of intervals: the $i$th interval for $i=0, \ldots, d$ is of length $2^i$. 

Consider the following input: for $i = 0, \ldots, d$, a job of size $2^i$ arrives (all jobs with release time zero, job $i+1$ follows job $i$), followed by $n$ jobs of length $1$ (where $n$ is determined later).
The cost for the optimal algorithm is:
\begin{eqnarray*}
	Cost(OPT)&=&\sum_{i=1}^{n}i+\sum_{i=0}^{d}\left(n+\sum_{j=0}^{i-1}2^j+2^i\right)\\
			&\leq& n^2+ (d+1)n+2\sum_{i=0}^{d}2^i\\
			&\leq& n^2+(d+1)n+2^{d+2}.
\end{eqnarray*}
In the proposed algorithm, the last $n$ unit-length jobs will be scheduled after the first $d+1$ jobs, which have total processing time of $2^{d+1}-1$. This implies:
\begin{eqnarray*}
Cost(ALG)\geq n2^d.
\end{eqnarray*}
For $n=\Theta(2^{\nicefrac{d}{2}})$, we get that $Cost(OPT)=\Theta(2^d)$, while $Cost(ALG)=\Omega\left(2^{d+\nicefrac{d}{2}}\right)$, leading to a competitive ratio of $\Omega\left( \sqrt{P_{\max}}\right)$.

The proposed algorithm failed because it did not leave enough space for the unit length jobs. 
A possible attempt to fix this problem would be to have more short intervals than long ones. 
One natural such sequence is (an infinite loop of) $2^d$ length 1 intervals, followed by $2^{d-1}$ length 2 intervals, etc., ending with a single interval of length $2^d$.

Consider an input sequence in which $2^{\nicefrac{d}{2}}$ jobs of size 2 arrive at time 0, followed by one large job of size $2^d$.
The optimal schedule processes the short jobs first, then the large one, resulting in cost: 
\begin{eqnarray*}
Cost(OPT)&=&\sum_{i=1}^{2^{\nicefrac{d}{2}}}2i+\left(2\cdot 2^{\nicefrac{d}{2}}+2^d\right)= \Theta\left(2^d\right).
\end{eqnarray*}
In the proposed algorithm, every short job will be scheduled after the first $2^d$ unit length intervals (as they do not fit unit length intervals). The obtained cost is thus 
\begin{eqnarray*}
Cost(ALG)\geq 2^d\cdot 2^{\nicefrac{d}{2}}=2^{d+\nicefrac{d}{2}},
\end{eqnarray*}
resulting in $\Omega\left(\sqrt{P_{\max}}\right)$ competitive ratio, as before.
Thus, saving too much space for short jobs might result in unnecessary delay, which may lead to a poor competitive ratio.

Motivated by the above two failed attempts, we now present our solution:

\subsection{The $S_k$ Integer and Interval Sequences } 
\label{sec:integer-sequence}

We define sequences of integers $S_k$, $k=0,1,\ldots$, as follows:  Let $S_0= \langle 1\rangle$ and for $k>0$ let $S_k=S_{k-1} \| S_{k-1} \| \langle 2^k\rangle$ where $\|$ denotes concatenation. Ergo,
\begin{eqnarray*} S_0 &=& \langle 1\rangle; \\
S_1 &=& S_0 \| S_0\| \langle 2^1\rangle = \langle 1,1,2\rangle ;\\
S_2 &=& S_1 \| S_1 \| \langle 2^2\rangle = \langle 1,1,2,1,1,2,4\rangle; \\
&\cdots&
\end{eqnarray*}

Let $n_k=2^{k+1}-1$ denote the length of $S_k$ (follows inductively from  $n_0=1$ and $n_k=2n_{k-1}+1$). Let $S_k[i]$, $i=1,\ldots,n_k$ be the $i$th element of $S_k$.  Let
$S_\infty$ be an infinite sequence whose length $n_k$ prefix is $S_k$ (for all $k$):
\[S_\infty=\langle1,1,2,1,1,2,4,1,1,2,1,1,2,4,8,1,1,2,1,1,2,4,1,1,2,1,1,2,4,8,16,1,\ldots\rangle.\] Let $S_\infty[i]$, $i=1,2, \ldots$ be the $i$th element of $S_\infty$. Note that $S_k[i] = S_{k'}[i]$ for all $k \leq k'$ and all $i= 1, \ldots, n_k$, ergo, $S_k$ is a prefix of $S_{k'}$ for $k\leq k'$.

\begin{lemma} \label{lem:sumsk} For all $d\geq 0$, for all $0 \leq k \leq d$, the sum of all the $2^k$ value items in $S_d$ is equal $2^d$:
	$$\sum_{1 \leq i \leq n_d: S_d[i]=2^k}2^k = 2^d.$$ \label{lem:len_sk}
\end{lemma}
\begin{proof}
	Proof via induction over $d$. The claim is obviously true for $S_0$. Assume the claim is true for $S_{d-1}$. {\sl I.e.}, for all $0 \leq k \leq d-1$, $$\sum_{1 \leq i \leq n_{d-1}: S_{d-1}[i]=2^k}2^k = 2^{d-1}.$$ Since $S_d$ is a concatenation of two $S_{d-1}$ sequences and the singleton sequence $\langle 2^d\rangle$, we get that for all $0 \leq k \leq d-1$ $$\sum_{1 \leq i \leq n_{d}: S_{d}[i]=2^k}2^k = 2\cdot 2^{d-1} = 2^d.$$ The claim also holds trivially for $k=d$.
\end{proof}

We use the  $S_k$ sequences to define interval sequences. 
Let $\gamma_i$ be the sum of the first $i$ entries in $S_\infty$, $\gamma_i = \sum_{j=1}^i S_{\infty}[j]$ ({\sl i.e.}, $\gamma_1=1$, $\gamma_2=2$, $\gamma_3=4$, etc.).

We define $S_k(t)$, $t\geq 0$, to be a sequence of $n_k$ consecutive intervals, the first of which starts at time $t$, and where the length of the $j$th interval equals $S_k[j]$. {\sl I.e.},
\[S_k(t)=\left\langle [t,t+\gamma_1], [t+\gamma_1,t+\gamma_2],
\ldots,\left[t+\gamma_{n_k-1},t+\gamma_{n_k}\right] \right\rangle.\]
For example
\begin{eqnarray}
S_2(2)= \langle [2,3],[3,4],[4,6],[6,7],[7,8],[8,10],[10,14]\rangle.\label{eq:s22}
\end{eqnarray}

For any interval sequence $S$ let $b(S)$ be the start of the first interval in $S$ and let $e(S)$ be the end of the last interval in $S$.  For example, $b(S_2(2))=2$ and $e(S_2(2))=14$. 

We say that $S_k$ {\sl appears} in $S_d(t)$ if there exists some $t'$ such that the interval sequence $S_k(t')$ is a contiguous subsequence of $S_d(t)$.
In this case we also say that $S_k(t')$ appears in $S_d(t)$.
Note that while $S_k$ is a sequence of integers, both  $S_k(t')$ and $S_d(t)$ are interval sequences.

By construction, for any $k$ and any $t\neq t'$ if $S_k(t)$ and $S_k(t')$ appear in some $S_d(\tilde{t})$, then $\left[b(S_k(t)),e(S_k(t))\right]$ and $\left[b(S_k(t')),e(S_k(t'))\right]$ are disjoint except, possibly, for their endpoints. Let $I$ be an interval of length $2^k$ that appears in $S_\infty(t)$. Then there is a unique $t'$ such that $S_k(t')$ appears in $S_\infty(t)$ and $I$ is the last interval of $S_k(t')$.
It follows from Lemma~\ref{lem:sumsk} that
\begin{corollary}	\label{cor:len_sk}
	 For all $k\leq d$, for all $t$,
	\begin{enumerate}
		\item $S_k$ appears in $S_d(t)$ $2^{d-k}$ times.
		\item The sum of the lengths of the intervals in $S_d(t)$ is $(d+1)2^d$.
	\end{enumerate}

\end{corollary}

The interval sequences defined above suggests a new possible static algorithm. Divide the timeline of each machine into intervals as in $S_\infty(0)$, and let any job that arrives occupy the first unoccupied interval it fits in. Unfortunately, as proved in section~\ref{sec:static_lb}, when the competitive ratio is evaluated as a function of $P_{\max}$ alone, this algorithm is $\Omega \left(\sqrt{P_{\max}}\right)$ competitive, as the natural algorithms in Section~\ref{sec:natural_algs}. (When the competitive ratio may be a function of $P_{\max}$ and $n$, this algorithm is $O\left(\log P_{\max}+\log n\right)$ competitive, see Theorem~\ref{thm:static_cr}).

\begin{definition}
A {\sl state} is a vector of consecutive interval sequences of the form
\begin{eqnarray*} A&=&\langle A_1 , A_2, \cdots , A_\ell\rangle \mbox{\rm\ where} \\
	A_i &=& S_{k_i}(t_i) \mbox{ for every $1\leq i \leq \ell$},\end{eqnarray*} for some $\ell$ (which we refer to as the \emph{length} of $A$) and integers $k_i$ for $1\leq i \leq \ell$, and where $e(A_i)=e(S_{k_i}(t_{i})) \leq t_{i+1}=b(A_{i+1})$ for $1 \leq i \leq \ell-1$. This means that the interval sequences are disjoint and ordered by their starting times. Note that there might be gaps between two consecutive state entries, i.e., $e(A_i) <b(A_{i+1})$ for some $1\leq i \leq \ell-1$.
\end{definition}

\subsection{\texorpdfstring{$O(\log P_{\max})$}{O(log Pmax)} Competitive Dynamic Menu }
\label{subsec:dynamic}

When job $j+1$ arrives the algorithm is in some {\sl configuration} $\psi^j=\left(A^{j},X^{j}\right)$, where $A^{j}$ is some state of length $\ell_{j}$, and $X^{j}$ is the set of intervals occupied by the previous $j$ jobs. State $A^j$ represents every machines' division of $\left[0,\max_{i\in[j]}c_i\right]$ into time intervals (same division for all machines). This division will be kept at any future time. 
For every $i<\ell_{j}$, $A^{j}_i$ is fixed and will be a part of every future state, while $A^{j}_{\ell_{j}}$ might be subject to change. We refer to $A^{j}_{\ell_{j}}$ as the {\sl tentative sequence} of state $A^j$. $X^j$ keeps track of all previously allocated intervals (in all machines): $([b,e],q)\in X^{j}$ means that some job $j'<j$ chose the interval $[b,e]$ on machine $1 \leq q \leq m$. Note that the size of job $j$, $p_j$, might be strictly smaller than the length of the interval ($e-b$), yet it is still considered occupied.

%
\noindent{\bf Generating the Dynamic Menu}

Given a state $A=\left(A_1,A_2,\ldots,A_\ell\right)$ and a time $t$, we define an interval sequence $\tau$ as follows:
\[
\tau(A,t) = \begin{cases}
	A_1\|A_2\|\ldots\|A_{\ell}\|S_\infty(t) & t\geq e\left(A_\ell\right) \\
	A_1\|A_2\|\ldots\|A_{\ell-1}\|S_\infty(b\left(A_\ell\right)) & t< e\left(A_\ell\right)
\end{cases}\]
$\tau$ is used to create the menu presented to a job $j$. We present an algorithm for the creation of the menu, based on the previous configuration $\psi^{j-1}$, and the current time $t$.
\newline
\MyFrame{
\begin{itemize}
	\item Let $\tau^j=\tau\left(A^{j-1},r_j\right)$.
	\item Set $d_1$ to be the length of the first time interval in $\tau^j$ beginning at time  $b_1\geq t$.
	\item Add $([b_1,b_1+d_1], q)$ to the menu for all machines $1\leq q\leq m$ in which $[b_1,b_1+d_1]$ is unoccupied (i.e, $([b_1,b_1+d_1],q)\notin X^{j-1}$).
	\item Set $i=1$
	
	\item Repeat until job $j$ chooses an interval:
	\begin{itemize}
		\item Let $d_{i+1}$ be the length of the first interval longer than $d_i$ in $\tau^j$ that starts at time $b_{i+1}\geq t$ (it follows that $b_{i+1}>b_i$).
			\item Add $([b_{i+1},b_{i+1}+d_{i+1}], q)$ to the menu for all machines $1\leq q\leq m$ in which $[b_{i+1},b_{i+1}+d_{i+1}]$ is unoccupied (i.e., $([b_{i+1},b_{i+1}+d_{i+1}],q)\notin X^{j-1}$).
		\item Set $i=i+1$.
	\end{itemize}
\end{itemize}}
\newline

By construction, no job will ever choose a time interval that starts before the job arrival time, nor will it ever choose a slot that has already been chosen.

A selfish job of length $p_j$ always chooses a menu entry of the form $([b,e],q$) where $b$ is the earliest menu entry with $p_j\leq e-b$.

%
%
%

\noindent{\bf Updating States.}

After job $j$ makes its choice of menu entry,  $(I(j),M(j))$, we update the configuration from $\psi^{j-1}=\left(A^{j-1},X^{j-1}\right)$ to $\psi^{j}=\left(A^{j},X^{j}\right)$. Clearly, $X^{j} = X^{j-1} \cup \left\{(I(j),M(j))\right\}$. In the rest of this section we describe how to compute $A^j$.

Recall that a state is a vector of consecutive and disjoint interval sequences.
Initially, $A^0=\langle\rangle$ with length $\ell_0=0$ and $A^0_{\ell_0}$ is an empty sequence with $b\left(A^0_{\ell_0}\right)=e\left(A^0_{\ell_0}\right)=0$.
$A^{j}$ always contains all of $A^{j-1}$'s interval sequences except possibly the tentative sequence $A^{j-1}_{\ell_{j-1}}$.
When job $j$ of size $2^k$ chooses an interval, the new tentative sequence $A^j_{\ell_j}$ can be one of the following:
\begin{enumerate}
	\item {\sl Unchanged from former:} The new tentative sequence in $A^j$ is the same as the former tentative sequence in $A^{j-1}$, {\sl i.e.}, $A^j_{\ell_j}=A^{j-1}_{\ell_{j-1}}$. This happens when $I(j)\in A^{j-1}$, see entry $1$ in Table~\ref{tab:states_table}.
	\item {\sl Disjoint from former:} The former tentative sequence, $A^{j-1}_{\ell_{j-1}}$ becomes {\sl fixed}, and the  new tentative sequence $A^j_{\ell_j}$ is disjoint from the former. The tentative sequence in $A^{j-1}$, $A^{j-1}_{\ell_{j-1}}$, is the $\ell_{j-1}$th element in all future states $A^i$, for $i\geq j$. See entries $2$ and $3$ in Table~\ref{tab:states_table}.
	\item {\sl Extension of former:} The new tentative sequence is an {\sl extension} of the former tentative sequence. {\sl I.e.}, if $A^{j-1}_{\ell_{j-1}}=S_d(t)$ then $\ell_j=\ell_{j-1}$ and $A^{j}_{\ell_{j}}=S_{k}(t)$. See entry $4$ in Table~\ref{tab:states_table}.
\end{enumerate}


\begin{table}[H]
\centering
\begin{tabu}{|[1pt] c|[1pt] c | c |[2pt] c | c | c |[1pt]}
	\tabucline[1pt]{-}
	 &  $\ell_j$			& $A^j_{\ell_j}$  & $c_j\leq e\left(A^{j-1}_{\ell_{j-1}}\right)$ &	$r_j\geq e\left(A^{j-1}_{\ell_{j-1}}\right)$ &
\begin{tabular}{c} $A_{\ell_{j-1}}^{j-1}=S_d(t)$ \\$k\leq d$\end{tabular} \\ \tabucline[1.5pt]{-}
	1&  $\ell_{j-1}$		& $A^{j-1}_{\ell_{j-1}}$ 			& True & - & - \\ \hline
	2& $\ell_{j-1}+1$ 	& $S_k(r_j)$ 			& False & True & - \\ \hline
	3&  $\ell_{j-1}+1$	& $S_k\left(e\left(A^{j-1}_{\ell_{j-1}}\right)\right)$ 	& False &  False & True \\ \hline
	4&  $\ell_{j-1}$		& $S_k\left(b\left(A^{j-1}_{\ell_{j-1}}\right)\right)$ 	& False& False & False \\
	\tabucline[1pt]{-}

\end{tabu}
\caption{Update rules: After job $j$ makes its choice (and $c_j$ is determined), the new state $A^j$ is a function of (i) $A^{j-1}$, (ii) release time $r_j$, (iii) processing time $p_j=2^k$, and (iv) completion time $c_j$.}
\label{tab:states_table}
\end{table}


Let $A^j_i,A^j_{i+1}$ be two consecutive interval sequences in a state $A^j$. If $b\left(A^j_{i+1}\right)>e\left(A^j_i\right)$, we say the interval $\left[e\left(A^j_i\right), b\left(A^j_{i+1}\right)\right]$ is a {\sl gap}.

Figure~\ref{fig:dyn_example} is an example with 5 jobs that arrive over time, and how the configuration changes over time. The jobs in Figure \ref{fig:dyn_example} illustrate cases 1--4 from Table \ref{tab:states_table} in the following order: case 2 for job 1, case 1 for job 2, case 3 for job 3, case 4 for job 4 and case 2 for job 5.



\begin{figure}
	\begin{center}
	\includegraphics[height=0.87\textheight,keepaspectratio]{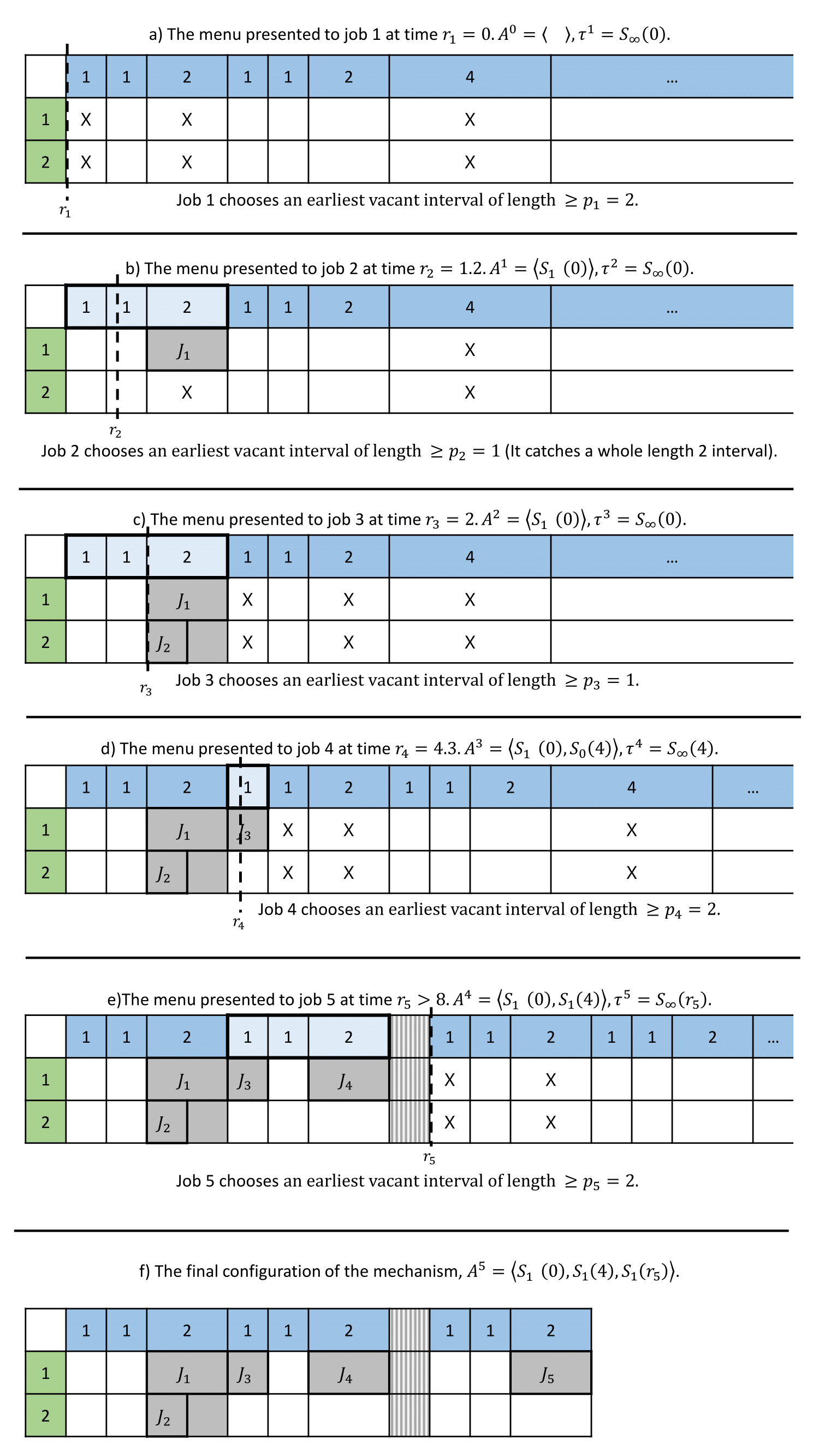}		
	\end{center}
	\caption{Changing Menus of the Dynamic Menu Algorithm, as jobs arrive and make choices. The two bottom rows in the tables represents two machines. An $X$ in a machine cell represents an (interval,machine) entry in the currently presented menu. A dashed line marks the release time of the current job. Gray cells represent choices previously made by jobs. A gap is represented by a rectangle filled with vertical lines. A rectangle outline in the top row of a table represents the tentative sequence before job $j$ makes it choice, i.e., $A^{j-1}_{\ell_{j-1}}$. Note that this example does not make the simplicity assumptions of Section~\ref{subsec:srpt}.}	
	\label{fig:dyn_example}
\end{figure}

Based on the definition of $A^j$ and its update rule, we observe the following.

\begin{observation}\label{obs:leaders} For every $A^j$,
	\begin{enumerate}
		\item If for job $j$, $c_{j}>e\left(A^{j-1}_{\ell_{j-1}}\right)$, then $I(j)$ is the last interval of the (new) tentative sequence $A_{\ell_{j}}^{j}$ which is of length $p_j$ .
		\item For every $A^j_i=S_{k_i}(t_i)$, there exists some job $j'\leq j$ such that $I(j')$ is the last interval in $S_{k_i}(t_i)$,
and $p_{j'}=2^{k_i}$. This means that job $j'$ occupies the entire last interval in $S_{k_i}(t_i)$ on machine $M(j')$.
	\end{enumerate}

\end{observation}	

\begin{proof}
	\begin{enumerate}
		\item $A^j$ must have been updated by one of the entries 2,3 or 4 in Table~\ref{tab:states_table}. In all these cases, the last interval in $A^j_{\ell_j}$ is of size $p_j$ and was chosen by job $j$ on some machine (follows from case analysis of the menu presented to job $j$ and its possible choices).
		\item For all $i$, $A^j_i$ was the tentative sequence in some past state $A^{j'}$ ($j'<j$). Let $\tilde{j}$ be the minimum $j'$ such that $A^j_i$ was the tentative sequence in state $A^{\tilde{j}}$. Then, $p_{\tilde{j}} = 2^{k_i}$.
	\end{enumerate}
\end{proof}

\subsection{Analysis}
\label{sec:analysis}

\subsubsection{Simplifying assumptions on the input sequence}\label{subsec:srpt}

For the purpose of analysis we assume an input sequence with integral release times and processing times that are powers of $2$.
When going from restricted inputs to the original inputs, the optimal preemptive algorithm cost improves by no more than a constant factor, whereas the online mechanism does not increase the sum of completion times.

Moreover, we assume that the input sequence never creates gaps as such gaps leave all machines free in both the online schedule and the optimal preemptive schedule (as a gap created by job $j$, implies jobs $1,\ldots,j-1$ were all fully processed by the online schedule before job j's arrival. Ergo, the optimal preemptive algorithm must also have completed processing jobs $1,\ldots,j-1$ prior to the arrival of job $j$). Therefore, a gap contribute equally to the sum of completion times of the online schedule and of the optimal preemptive schedule. This improves the competitive ratio. Therefore, an adversary generating such a sequence will never introduce gaps, {\sl  i.e.}, for any state $A^j$, $e\left(A^j_i\right)=b\left(A^j_{i+1}\right)$ for every $1\leq i<\ell_j$.

\subsubsection{Comparison to SRPT}

We now turn to analyze the performance of our mechanism. This is done by comparing the completion time of each job in our mechanism and in SRPT.
Let $j$ be a job in the input sequence. We define $D(j)=\left\{j'\leq j | p_{j'}\leq p_j\right\}$ to be the set of all jobs that arrived no later than job $j$ and that are no bigger than it (note $j\in D(j)$). These jobs are all completed no later than job $j$ both in our mechanism and in SRPT, i.e., $c^*_j\geq \frac{1}{m}vol(D(j))$ (where $c^*_j$ is the completion time of job $j$ in SRPT).  Our analysis is based on this set.

We start with a with a few simple properties:
Recall that the last interval in $S_d(t)$ is the only interval of size $2^d$ in the sequence. The following lemma
gives a lower bound on the completion time of a job (in the optimal schedule) that chose the last interval in $S_d(t)$ (under the dynamic menu mechanism).

\begin{lemma}\label{lem:caught_by_smaller}
	Let $d\geq 0$. Let $t$ and $q$ be such that some job $j$ with $r_j\leq t, p_j=2^k$ chose the last interval of $S_d(t)$ on machine $q$ ($2^k\leq 2^d$).
	Let $D(j,q,S_d(t))$ be the set of jobs, completing no later than job $j$ under SRPT, that execute on the same machine as job $j$, and occupy some interval in $S_d(t)$. {\sl I.e.,}  $D(j,q,S_d(t))=D(j)\cap \left\{j'|I(j')\in S_d(t), M(j')=q\right\}$. Note that $j\in D(j,q,S_d(t))$.
	Then, $$vol\left(D(j,S_d(t),q)\right) \geq 2^d.$$
\end{lemma}
\begin{proof}
	
	If $p_{j}=2^k=2^d$ then as $j \in D(j,S_d(t),q)$ the claim is clearly true.
	Specifically, the claim is true for $d=0$ since all processing times are $\geq 1$ so $p_{j}=1=2^d$.
	It remains to consider the case where $p_{j}=2^k<2^d$ ({\sl i.e.}, $k<d$).
	
	Proof via induction over $d$. The claim is true for $d=0$ as stated above.
	
	Let $d>0$, and assume the claim is true for all $0 \leq d' <d$.
	
	If $k<d$ then it must be the case that every interval of length $2^k$ in $S_d(t)$ is occupied on machine $q$ (by a job $j'<j$), otherwise job $j$ would have preferred such an interval over its choice. 
	Let $I$ be an interval of length $2^k$ in $S_d(t)$. By our construction, $I$ is the last interval of some (unique) $S_k$ appearance in $S_d(t)$, i.e., $S_k(t')$ for $t'\geq t$ and $e(S_k(t'))< e(S_d(t))$. $(I,q)$ is occupied by some job $j'<j$ (see top row in Figure \ref{fig:lemmacbs_example}). It follows that $r_{j'}\leq r_j\leq t\leq t'$. By the induction hypothesis, $vol(D(j',S_k(t'),q))\geq 2^k$.
	
	It now follows from Lemma~\ref{lem:len_sk}, that $S_k$ appears $2^{d-k}$ times in $S_d(t)$ (irrespective of $t$), so we can conclude that $vol(D(j,S_d(t),q))\geq 2^d$, as desired.
\end{proof}

\begin{corollary}\label{cor:caught_by_smaller_late_release}
For $d\geq 1$, replacing the condition that $r_j\leq t$ in Lemma \ref{lem:caught_by_smaller} above with the condition $r_j\leq e(S_{d-1}(t))$, gives a [weaker] guarantee that   $vol\left(D(j,S_d(t),q)\right) \geq2^{d-1}$.
\end{corollary}
\begin{proof}
Let $j$, $d$, and $t$ be as in Lemma \ref{lem:caught_by_smaller}.
	If $p_{j}=2^k=2^d$ then the claim is true. Otherwise, $2^k\leq 2^{d-1}$. Recall that by construction, $S_d(t)=S_{d-1}(t)\|S_{d-1}(t_2)\|\langle I\rangle $ where $t_2=e(S_{d-1}(t))$ and $I=[e(S_{d-1}(t_2)), e(S_{d-1}(t_2))+2^d]$ is a length $2^d$ interval. The last interval in $S_{d-1}(t_2)$ is of length $2^{d-1}$ and must be occupied on every machine when job $j$ arrived (otherwise it would have chose it on some available machine). Thus, it must by occupied by some job $j'$ with $r_{j'}\leq r_{j}\leq e(S_{d-1}(t))=b(S_{d-1}(t_2)$ and $p_{j'}\leq 2^{d-1}$. Applying Lemma~\ref{lem:caught_by_smaller} to job $j'$ and  $S_{d-1}(t_2)$ gives the desired result.
\end{proof}

\begin{figure}
	\fbox{\includegraphics[width=0.9\textwidth,keepaspectratio]{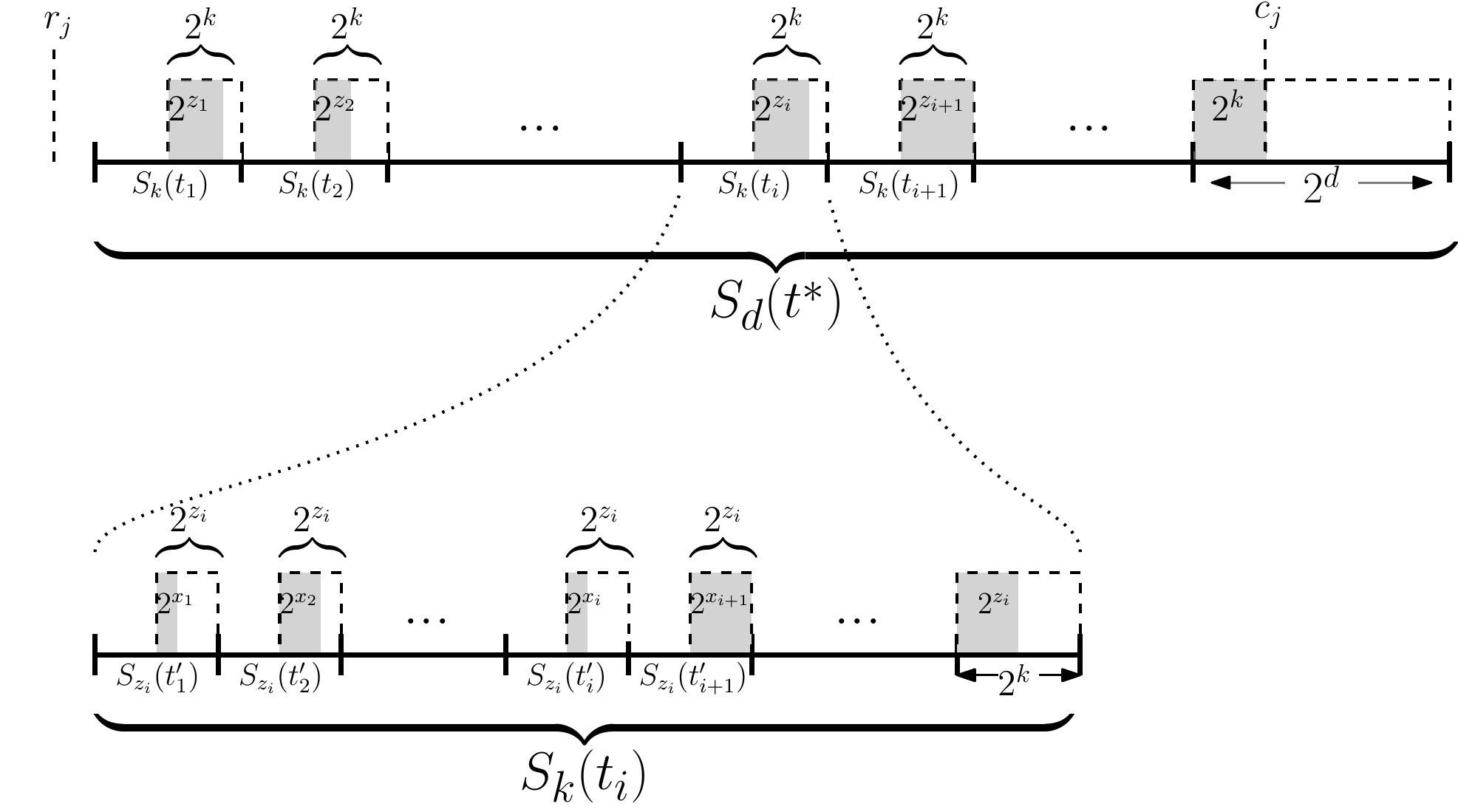}}
	\caption{Illustration for the proof of Lemma \ref{lem:caught_by_smaller}. The upper figure shows the last interval in $S_d(t)$ is occupied by a size $2^k$ job, which implies all $2^k$ intervals in $S_d(t)$ have already been occupied. One of the $S_k$ appearances in $S_d(t)$ is expended in the lower figure, illustrating a structure similar to the upper figure. It follows inductively that the sum of the lengths of jobs in $S_k(t_i)$ of length $\leq 2^k$ is $\geq 2^k$. }
	\label{fig:lemmacbs_example}
\end{figure}

For any job $j$, let $c^*_j$ be the completion time of job $j$ in the SRPT schedule.
Our goal is to show that  $c_j \leq \log P_{\max} \cdot c^*_j$

Consider job  $j$, and the final state $A^n$. Let $a_j$, $\rho_j$ be such that $$b\left(A^n_{a_j}\right) \leq r_j <e\left(A^n_{a_j}\right) \mbox{\rm\ and\ }
b\left(A^n_{a_j+\rho_j}\right) < c_j \leq e\left(A^n_{a_j+\rho_j}\right)\qquad  (\rho_j \geq 0).$$
\begin{figure}
	\fbox{\includegraphics[width=0.9\textwidth,keepaspectratio]{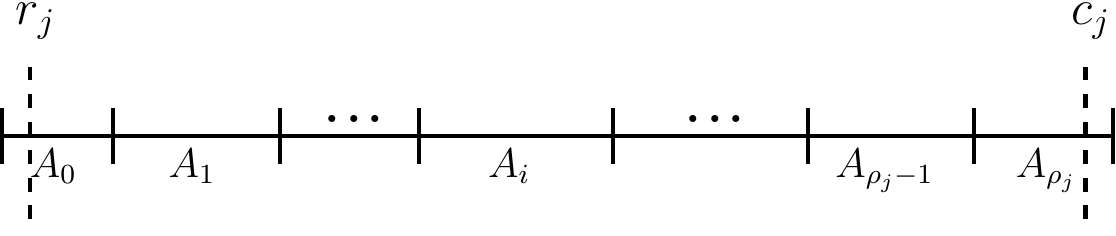}}
	\caption{The timeline division induced by the final state $A^n$.}
	\label{fig:states_example}
\end{figure}

To avoid repeatedly using cumbersome notation, we use the shorthand \[A_0=A^n_{a_j},A_1=A^n_{a_j+1},\ldots, A_{\rho_j}=A^n_{a_j+\rho_j}\] (see Figure~\ref{fig:states_example}).


%

\begin{lemma}\label{lem:all_machines_caught}
	Assume $\rho_j>0$. Let $p_j=2^k$ be the size of job $j$. Let $A_i=S_d(t)$ for some $i<\rho_j$ and some $d,t$. If $r_j\leq b\left(A_i\right)$, then $vol(D(j)\cap \left\{j'|I(j')\in A_i\right\})\geq m\cdot 2^d$, i.e., the volume of jobs $j'\leq j$ of size $p_{j'} \leq 2^k$ and for which $I(j')\in A_i$ is at least $m\cdot 2^d$.
\end{lemma}
\begin{proof}
	We separate the proof into two cases:
	\begin{itemize}
		\item 	If $d\geq k$, then, on all machines, the last interval of every $S_k$ appearance in $S_d(t)$ is occupied by a job that arrived before job $j$ (by a job of size $\leq 2^k$). Otherwise, job $j$ would have chosen such an unoccupied interval. By using Lemma~\ref{lem:caught_by_smaller} on every one of the $S_k$ appearances (on every machine, separately) 
		we get that for every such $S_k(t')$ appearance in $S_d(t)$ (on every machine), $vol(D(j)\cap \left\{j'|I(j')\in S_k(t') \right\})\geq 2^k$, i.e., every such appearance has a volume of at least $2^k$ of jobs in $\left\{ j'\leq j| p_{j'}\leq 2^k, I(j') \in S_k(t') \right\}$. 
		Taken together with Corollary~\ref{cor:len_sk} this implies that $vol(D(j)\cap \left\{j'|I(j')\in S_d(t) \right\})\geq m\cdot 2^d$.
		\item 	Otherwise ($d<k$),  consider the minimal index $j'$ with $c_{j'}>e\left(A_i\right)$; this means that the tentative sequence $A^{j'}_{\ell_{j'}}$ is disjoint from the tentative sequence $A^{j'-1}_{\ell_{j'-1}}$. (If $A^{j'}_{\ell_{j'}}$ was unchanged, or an extension of $A^{j'-1}_{\ell_{j'-1}}$ it contradicts the assumption that  $c_{j'}> e\left(A_i\right)$.

		Since $c_j>e\left(A_i\right)$, $j'$ arrived no later than job $j$ and $r_{j'}\leq r_j$.
		Let $p_{j'}=2^z$.	It must be the case that $z\leq d$, otherwise, $A^{j'}_{\ell_{j'}}$ would be an extension of $A^{j'-1}_{\ell_{j'-1}}$.

		
		Therefore, since $j'$ can fit in the last interval of $A_i=S_d(t)$ and its release time is no later than $b\left(A_i\right)$, it must be the case that this interval is occupied on every machine, by some job of size $\leq 2^d\leq 2^k$ that arrived \textit{before} $j'$. By applying
		Lemma~\ref{lem:caught_by_smaller}, 
		we get that in every machine, there are jobs of size $\leq 2^k$ in $A_i$ that arrived no later than $j'$ of volume at least $2^d$.
		we get that in every machine, $vol(D(j')\cap \left\{\tilde{j}|I(\tilde{j})\in A_i\right\})\geq 2^d$, i.e., there are jobs of size $\leq p{j'}$ in $A_i$ that arrived no later than $j'$ of volume at least $2^d$.
		Since $j'$ arrived no later than $j$, and $p_{j'}\leq p_j$, the lemma follows.
	\end{itemize}
\end{proof}

\begin{figure}
	\fbox{\includegraphics[width=0.9\textwidth,keepaspectratio]{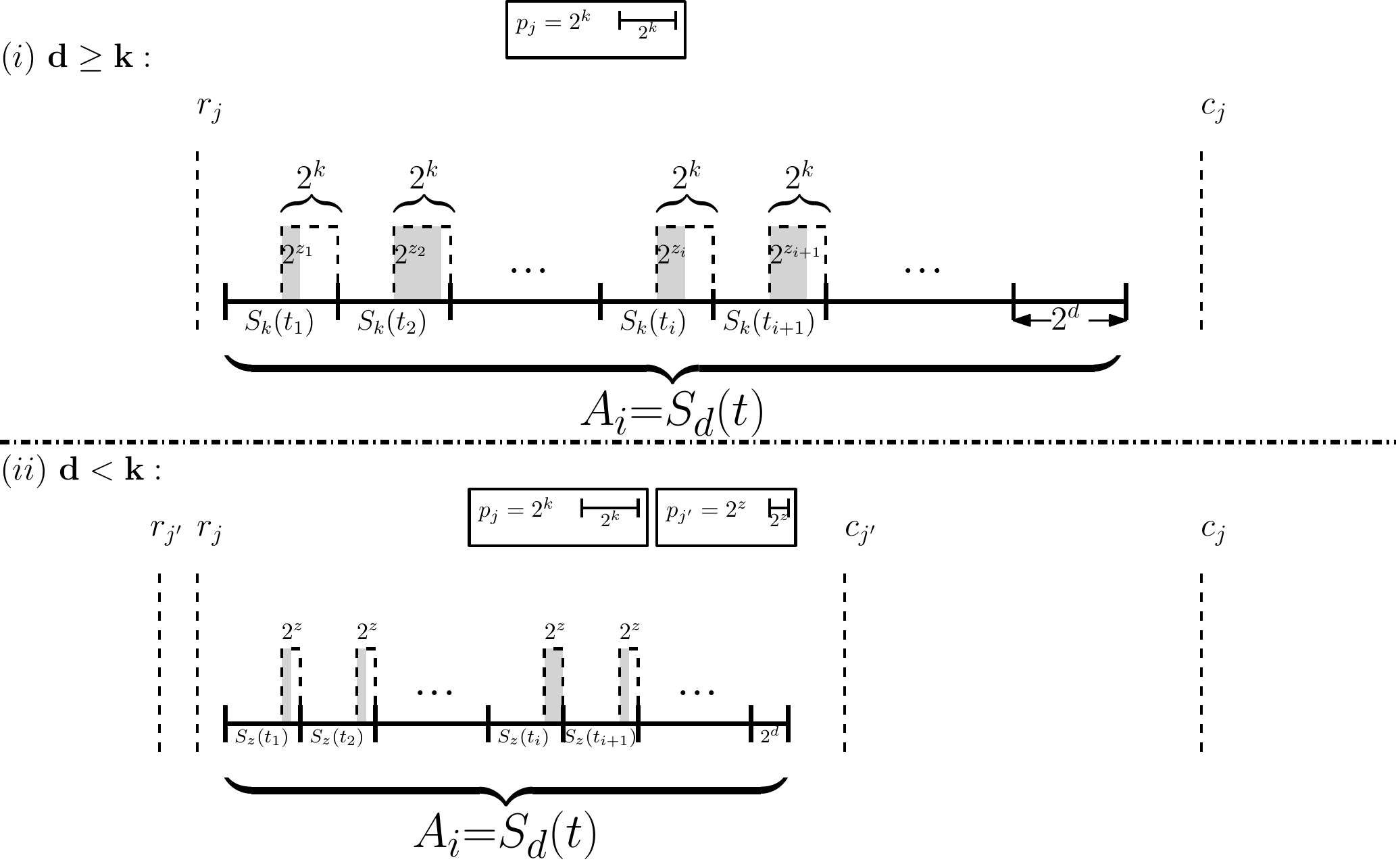}}
	\caption{Illustration for the proof of Lemma \ref{lem:all_machines_caught}. Job $j$ is size $p_j=2^k$, $r_j\leq b\left(A_i\right)$ and $c_j>e\left(A_i\right)$. If $d\geq k$ (see upper figure) then, on all machines, all last intervals in $S_k$ appearances in $S_d(t)$ are occupied.  Lemma~\ref{lem:caught_by_smaller} implies that every such appearance contains a set of small jobs of large volume ($\geq 2^k$). If $d<k$ (the lower figure) then $j'$ is a job if size $p_{j'}=2^z$, $r_{j'}\leq b\left(A_i\right)$ and $c_{j'}>e\left(A_i\right)$, the same arguments for job $j'$ give the desired result.}
	\label{fig:lemmacbs_example}
\end{figure}

The next lemma handles the case where job $j$ arrives sufficiently early in $A_0$.

\begin{lemma}\label{lem:all_machines_almost_caught}
	Let $\rho_j >0$, $p_j=2^k$ be the size of job $j$ and let $A_0=S_d(t)$. If $0<d\leq k$ and $r_j\leq e(S_{d-1}(t))$, then 
	$vol(D(j)\cap \left\{j'|I(j')\in A_0 \right\})\geq m\cdot 2^{d-1}$.
\end{lemma}
\begin{proof}
	Consider the minimal index $j'$ with $c_{j'}>e\left(A_0\right)$; that is, $A_0$ was the tentative sequence before job $j'$ arrived, and became fixed afterwards (by Observation~\ref{obs:leaders}). Since $c_j>e\left(A_0\right)$, $j'$ arrived no later than job $j$ and $r_{j'}\leq r_j$. Let $p_{j'}=2^z$. It must be the case that $z\leq d$, otherwise, $A_0$ would have not been fixed by job $j'$, but would have been extended to an $S_z(t)$ sequence (and possibly an $S_y(t)$ sequence for some $y>z$ in a later stage).
	
	Therefore, since $j'$ can fit in the last interval of $A_i=S_d(t)$ and its release time is no later than $e(S_{d-1}(t))$, it must be the case that this interval is taken \textit{in every machine}, and it must be occupied by some job of size $\leq 2^d\leq 2^k$ that arrived \textit{before} $j'$. By applying Corollary~\ref{cor:caught_by_smaller_late_release} we get that for every machine $1\leq q\leq m$, $vol\left(D(j')\cap \left\{\tilde{j}|I(\tilde{j})\in S_d(t), M(\tilde{j})=q \right\}\right)\geq 2^{d-1}$, and the lemma follows.

%
\end{proof}

Let $c^*_j$ be the completion time of job $j$ in SRPT.
\begin{lemma}\label{lem:dyn_proof1}
	For every $j$ such that $\rho_j =0$ we have that $c_j=O(\log P_{\max})\cdot c^*_j$.
\end{lemma}
\begin{proof}
	$\rho_j=0$ means that $b\left(A_0\right)\leq r_j <c_j \leq e\left(A_0\right)$.
	Let $A_0=S_d(t)$, as $j$ fits in this sequence, $k\leq d$. We consider the following cases.
	\begin{enumerate}[a.]
		\item If $j$ occupies the first interval of length $2^k$ in $S_d(t)$, then it is scheduled in $S_k(t)$, and by Corollary~\ref{cor:len_sk} $c_j \leq b\left(A_0\right)+(k+1)2^k$ and $c^*_j\geq b\left(A_0\right)+2^k$. Thus, $$c_j\leq (k+1)c^*_j\leq (\log P_{\max}+1)c^*_j$$ as desired.
		\item Otherwise, let $a$ be such that $e(S_a(t))\leq r_j+p_j< e(S_{a+1}(t))$ (note that $r_j+p_j \geq 1$ implies that there exists such $a$). If $c_j\leq e(S_{a+2}(t))$ then, by Corollary~\ref{cor:len_sk}, $c_j \leq b\left(A_0\right)+(a+3)2^{a+2}$ and $c^*_j\geq b\left(A_0\right)+(a+1)2^a$. We get that in this case
		$$c_j\leq \frac{4(a+3)}{a+1}\cdot c^*_j\leq 12 c^*_j.$$
		
		\item 	Finally, consider the case where $e(S_{a+x}(t))<c_j\leq e(S_{a+x+1}(t))$ for some $x>1$. When $j$ arrives, for all machines $1 \leq q \leq m$ and for all intervals $I$, of length $\geq 2^k$, where $I\in S_{a+x}(t)\setminus S_{a+1}(t)$, $(I, q)$ is occupied (otherwise, $j$ would have preferred such $(I, q)$ over his choice). Recall that by Corollary~\ref{cor:len_sk}, $S_k$ appears $2^{a+x-k}$ times in $S_{a+x}(t)$ , and $S_k$ appears at most $2^{a+1-k}$ times in $S_{a+1}(t)$ (only if $a+1\geq k$). Therefore, $S_k$ appears at least $2^{a+x-k}-2^{a+1-k}\geq 2^{a+x-k-1}$ times in $S_{a+x}(t)\setminus S_{a+1}(t)$ \textit{on each machine}.
		
		For each machine $q$ and each such $S_k$ appearance, we apply Lemma \ref{lem:caught_by_smaller}. We deduce that each such appearance, $S_k(t')$ for some $t'$,  contains a $2^k$ volume of jobs in $\left\{j'\leq j|p_{j'}\leq 2^k, I(j')\in S_k(t')\right\}$. 
		Thus, the total volume of jobs that must be processed before $j$ on SRPT is at least $m\cdot 2^{a+x-1}$, and therefore, $c^*_j\geq 2^{a+x-1}$. Also, it must be the case that $c^*_j\geq b\left(A_0\right)$ and therefore $c^*_j\geq \max\left\{b\left(A_0\right),2^{a+x-1}\right\}$.
		
		On the other hand, by Corollary~\ref{cor:len_sk}, $$c_j\leq e(S_{a+x+1}(t)) = b\left(A_0\right)+(a+x+2)2^{a+x+1}\leq \max\left\{b\left(A_0\right),2^{a+x-1}\right\}(5+4(a+x+1)).$$
		By Observation~\ref{obs:leaders}, the last interval in $A_0$ is taken by a job of size $2^d\geq 2^{a+x+1}$, and therefore, $\log P_{\max}\geq a+x+1$. We get that
		$$c_j\leq (5+4\log P_{\max})\cdot c^*_j$$
		 as desired. 
	\end{enumerate}
\end{proof}

\begin{lemma}\label{lem:dyn_proof2}
	If $\rho_j\geq 1$, then for every $j$, $c_j=O(\log P_{\max})\cdot c^*_j$.
\end{lemma}
\begin{proof}
	Recall that $D(j)$ is the set of jobs that are completed no later than job $j$ in both SRPT and ALG. Let $D_i$ be the jobs $j'\in D(j)$ that chose some time interval in $A_i$ on some machine (notice that $\underset{0\leq i\leq \rho_j}{\cup} D_i\subseteq D(j)$).
	As $c^*_j$ cannot be smaller than $r_j+p_j$, it follows that $$c^*_j\geq \max\left\{r_j+p_j,\ \frac{1}{m}\sum_{i=0}^{\rho_j}vol(D_i)\right\}.$$
	For every $D_i$, we introduce a lower bound $\gamma_i$ with the property that $\gamma_i\leq vol(D_i)$. Therefore,
	\begin{eqnarray}
		c^*_j\geq \max\left\{r_j+p_j,\ \frac{1}{m}\sum_{i=0}^{\rho_j}\gamma_i\right\}.\label{eq:srpt_lb}
	\end{eqnarray}	
		We now consider the following case analysis for $A_i$:
		\begin{enumerate}[a.]
			\item $i=0$, given that $r_j>b\left(A_0\right)$.
			
			Let $A_0=S_d(t)$, and let $a$ be the maximal integer such that $e(S_a(t))\leq r_j <e(S_{a+1}(t))$ (recall that $r_j<e\left(A_0\right)$, thus $d\geq a+1)$. By Observation~\ref{obs:leaders}, $P_{\max}\geq 2^d$, and $\log P_{\max}\geq d$.

			\begin{enumerate}[i.]
				\item If $d=a+1$ then set $\gamma_0=0$. It follows from Corollary~\ref{cor:len_sk} that $r_j\geq b\left(A_0\right)+(a+1)2^a$ and thus $$e\left(A_0\right)=b\left(A_0\right)+(d+1)2^d= b\left(A_0\right)+(a+2)2^{a+1}\leq 4r_j$$.
				\item Otherwise, $d\geq a+2$, and we set $\gamma_0=m2^{d-1}$. It follows from Lemma~\ref{lem:all_machines_almost_caught} that $\gamma_0\leq vol(D_0)$. Since $e\left(A_0\right)=b\left(A_0\right)+(d+1)2^d$,  $e\left(A_0\right)\leq b\left(A_0\right)+2(\log P_{\max}+1)\frac{\gamma_0}{m}$.
			\end{enumerate}
	
			In both these cases,
			\begin{equation}
			e\left(A_0\right)-b\left(A_0\right)\leq \max\left\{4r_j,2(\log P_{\max}+1)\frac{\gamma_0}{m}\right\}\leq 4(\log P_{\max}+1)\max\left\{r_j,\frac{1}{m}\gamma_0\right\}.\label{eq3}
			\end{equation}
			
			\item $i\in \left\{1,\ldots,\rho_j - 1\right\}$, or $i=0$ given that $r_j=b\left(A_0\right)$.
			
			It follows that $r_j\leq b\left(A_i\right)<e\left(A_i\right)<c_j$. Let $A_i=S_d(t)$. By Observation~\ref{obs:leaders} $P_{\max}>d$. Set $\gamma_i=m2^d$. By Lemma~\ref{lem:all_machines_caught} $\gamma_i\leq vol(D_i)$. Thus,
			\begin{equation}
			e\left(A_i\right)-b\left(A_i\right)=(d+1)2^d\leq (\log P_{\max}+1)\frac{1}{m}\gamma_i.\label{eq4}
			\end{equation}
			
			\item $i=\rho_j$.
			
			Let $A_{\rho_j}=S_d(t)$. If $j$ chooses the first $2^k$ interval in $A_{\rho_j}$ then we set $\gamma_{\rho_j}=2^k\leq vol(d_{\rho_j})$ and $j$'s completion time is $c_j=b\left(A_{\rho_j}\right)+(k+1)2^k\leq b\left(A_{\rho_j}\right)+(\log P_{\max}+1)p_j$. Otherwise, let $a$ be the maximal integer such that $e(S_a(t))<c_j\leq e(S_{a+1}(t))$ ($d\geq a+1\geq k+1$). Every $2^k$ interval in $S_a(t)$ is occupied by a job that arrived before $j$, thus by Lemma~\ref{lem:caught_by_smaller} and Corollary~\ref{cor:len_sk}, $\gamma_{\rho_j} = 2^a\leq vol(D_{\rho_j})$, and $c_j\leq b\left(A_{\rho_j}\right)+(a+2)2^{a+1}\leq b\left(A_{\rho_j}\right)+(\log P_{\max}+1)\frac{\gamma_{\rho_j}}{m}$. In any case,
			\begin{equation}
			c_j-b\left(A_{\rho_j}\right) \leq (\log P_{\max}+1)\max\left\{p_j,\frac{1}{m}\gamma_{\rho_j}\right\}.\label{eq5}
			\end{equation}
			
		\end{enumerate}
		By the assumption there are no gaps in the schedule, $e\left(A_i\right)=b\left(A_{i+1}\right)$. Since $r_j\geq b\left(A_0\right)$ by definition, and by \eqref{eq3}, \eqref{eq4} and \eqref{eq5},
		\begin{eqnarray*}
		c_j &=& b\left(A_0\right)+\left(e\left(A_0\right)-b\left(A_0\right)\right)+\sum_{i=1}^{\rho_j-1}\left(e\left(A_i\right)-b\left(A_i\right)\right)+ \left(c_j-b\left(A_{\rho_j}\right)\right)\\
		&\leq& r_j+ 4\left(\log P_{\max}+1\right)\max\left\{r_j,\frac{1}{m}\gamma_0\right\} + \frac{1}{m} \sum_{i=1}^{\rho_j-1}(\log P_{\max}+1)\gamma_i + (\log P_{\max}+1)\max\left\{p_j,\frac{1}{m}\gamma_{\rho_j}\right\}\\
		& = & r_j +(\log P_{\max}+1)\left( 4\max\left\{r_j,\frac{1}{m}\gamma_0\right\} + \frac{1}{m} \sum_{i=1}^{\rho_j-1}\gamma_i + \max\left\{p_j,\frac{1}{m}\gamma_{\rho_j}\right\}\right)\\
		& \leq & c^*_j +(\log P_{\max}+1)\left( 4c^*_j + c^*_j + c^*_j\right)\\
		& = & O(\log P_{\max})c^*_j,
	\end{eqnarray*}
	where the second inequality follows from ~\eqref{eq:srpt_lb}. This concludes the proof of the lemma.

\end{proof}
Lemmata \ref{lem:dyn_proof1} and \ref{lem:dyn_proof2} above imply the following theorem:
\begin{theorem}
	The mechanism presented in this section is $O(\log P_{\max})$ competitive.
\end{theorem}

\subsection{Arbitrary processing time, Weight $\leq B_{\max}$}

The static algorithm suggested at the end of Section~\ref{sec:integer-sequence} used for weight one jobs of arbitrary length can be easily adapted to
	weights in some predetermined range from $1$ to $B_{\max}$. Replicate every interval in the sequence $S_\infty(0)$ $\log B_{\max}+1$ times. For $\ell=0,\ldots, \log B_{\max}$, the $\ell$th copy is designed to hold only jobs of weight $\geq 2^{\ell}$.
	To achieve this, one associates prices with such intervals, as done in Section~\ref{sec:pricing-weighted}.
	The analysis preformed in  Appendix~\ref{thm:static_cr} holds when multiplying every element with $\log B_{\max}+1$,  implying a competitive ratio of $O\left((\log P_{\max}+\log n_{\max}) \cdot \log B_{\max}\right)$.

\section{Lower Bound on the Competitive Ratio for any Prompt Online Algorithm, Arbitrary Lengths}
\newcommand{\p}{\pmb{P}}
We now show that any prompt online scheduling algorithm must have a competitive ratio of  $\Omega\left(\log   P_{\max}\right)$, even
if randomization is allowed.

Let $c$ be the competitive ratio of some algorithm ALG as a function of $ P_{\max}$.
Consider the following sequence, for $\p$ to be determined later:\newline

\MyFrame{
	For $i=0,\ldots, 16c$:
	\begin{itemize}
		\item $n_i=2^i$ jobs of size $P_i=\frac{\p}{2^i}$ arrive one after the other (at time $0$).
		
		\item If the expected number of $P_i$ sized jobs with completion time greater than $8c\p$ is at least $n_i/2$, stop the sequence. Let $j$ be the last iteration.
	\end{itemize}
}
\newline\newline
Note that for every $i=0,\ldots, 16c$ it holds that $n_i\cdot P_i=\p$.

\begin{lemma}
	There must be an iteration $j\in \{0,\ldots ,16c\}$ for which in expectation more than half of the jobs have completion time greater than $8c\p$. \label{lem:overflow}
\end{lemma}
\begin{proof}
	Let $X_i$ be a random variable representing the number of size $P_i$ jobs, with completion time greater than $8c\p$. If for all $i\in \{0,\ldots, 16c\}$, $\expect{}{X_i}\leq n_i/2$, then the total expected volume of jobs completed before time $8c\p$ is at least
	\[
	\sum_{i=0}^{16c} \expect{}{\left(n_i-X_i\right)\cdot P_i}  \geq \sum_{i=0}^{16c} \frac{n_i}{2}\cdot P_i = \sum_{i=0}^{16c}\frac{\p}{2} > 8c\p,
	\]
	a contradiction.
\end{proof}

\begin{theorem}
	Any random prompt online algorithm must be $\Omega\left(\log P_{\max}\right)$  competitive for the above sequence.
\end{theorem}

\begin{proof}	
	According to Lemma \ref{lem:overflow}, there must be some $j\in \{0,\ldots, 16c\}$ for which in expectation at least half of the jobs are completed after time $8c\p$. Given this $j$, we give bounds on both $OPT$ and ALG. Let $X_i$ be as in Lemma~\ref{lem:overflow}. In ALG, $\expect{}{X_j}> n_j/2$, thus:
	
	\begin{eqnarray}
	\expect{}{\mbox{\rm Cost}(\mbox{\rm ALG})}> \expect{}{X_j\cdot 8c\p}> 8c\p\cdot \frac{n_j}{2} = 4c\p\cdot n_j.\label{eq:alg_lb}
	\end{eqnarray}
	In OPT, the jobs are scheduled from the smallest one (of size $P_j$) to the biggest one (of size $P_0=\p$). The $k$th job of size $P_i$ to  be scheduled, is completed after all jobs smaller than it (of sizes $P_{i+1},\ldots,P_j$) and after $k-1$ jobs of size $P_i$, and therefore has a completion time of $$\left(\sum_{\ell=i+1}^{j}n_\ell\cdot P_\ell \right) + P_i\cdot (k-1)+P_i=P_i\cdot k + \sum_{\ell=i+1}^{j}n_\ell\cdot P_\ell.$$
	Summing over all jobs of all sizes, we have
	\begin{eqnarray}
	\mbox{\rm Cost}(\mbox{\rm OPT}) & = & \sum_{i=0}^j\left(\sum_{k=1}^{n_i}\left(P_i\cdot k + \sum_{\ell=i+1}^{j}n_\ell\cdot P_\ell\right)\right) \nonumber\\
	& = & \underbrace{\sum_{k=1}^{n_j} P_j\cdot k}_{\text{$(i)$}} + \underbrace{\sum_{i=0}^{j-1}\sum_{k=1}^{n_i}P_i\cdot k}_{\text{$(ii)$}} +   \underbrace{\sum_{i=0}^{j-1}\left(\sum_{\ell=i+1}^{j}n_\ell\cdot P_\ell\right)\cdot n_i}_{\text{$(iii)$}}.
	\end{eqnarray}
	We now bound each term of $\mbox{\rm Cost}(\mbox{\rm OPT})$ separately.
	\begin{eqnarray}
	(i): P_j\sum_{k=1}^{n_j}  k < P_j\cdot n_j^2 = \p\cdot n_j. \label{eq:bound1}
	\end{eqnarray}

	\begin{eqnarray}
	(ii) :  \sum_{i=0}^{j-1}P_i\sum_{k=1}^{n_i} k < \sum_{i=0}^{j-1} P_i\cdot n_i^2 = P \cdot \sum_{i=0}^{j-1}2^i \leq \p \cdot 2^j = \p\cdot n_j. \label{eq:bound2}
	\end{eqnarray}
	
	For $(iii)$ we have
	\begin{eqnarray}
	(iii) : \sum_{i=0}^{j-1}\left(\sum_{\ell=i+1}^{j}n_\ell\cdot P_\ell\right)\cdot n_i & = &  \sum_{i=0}^{j-1}\sum_{\ell=i+1}^{j} P\cdot 2^i \nonumber\\
	& = & P\sum_{i=0}^{j-1}(j-i)2^i \nonumber \\
	& = & P\sum_{i=1}^{j}i\cdot 2^{j-i} \nonumber \\
	& = & P\cdot 2^j \sum_{i=1}^{j}\frac{i}{2^i} \nonumber\\
	& \leq & 2P\cdot n_j.\label{eq:bound3}
	\end{eqnarray}
	
	From Equations \eqref{eq:bound1}, \eqref{eq:bound2} and \eqref{eq:bound3}, we get that $\mbox{\rm Cost}(\mbox{\rm OPT}) \leq 4\p\cdot n_j$. Therefore, $$\expect{}{\mbox{\rm Cost}(\mbox{\rm ALG})/\mbox{\rm Cost}(\mbox{\rm OPT})}> c,$$ in contradiction to the assumption that ALG is $c$-competitive.
	
	For the input sequence to be valid, it must be that $P_j\geq 1$. As $j\leq 16c$, it is sufficient that $c\left(\p\right)\leq \frac{1}{16}\log{\p}$, as in this case, $P_{16c} = \frac{\p}{2^{16c}}\geq 1$. So for every competitive ratio function $c$ such that $c\left(P_{\max}\right)=o\left(\log P_{\max}\right)$ there exists a sufficiently large $\p$ for which $c\left(\p\right)\leq \frac{1}{16}\log{\p}$, and our input is a valid counter example.
\end{proof}

\section{Pricing Time Slots for Selfish Weighted Jobs}
\label{sec:pricing-weighted}

\newcommand{\len}[1]{\left|#1\right|}
\newcommand{\series}[1]{\langle#1\rangle}
\newcommand{\floor}[1]{\lfloor#1\rfloor}
\newcommand{\wmax}{\log W_{\max}}

In this section we introduce a dynamic menu based mechanism, for unit length jobs (unit processing time). The resulting competitive ratio is $O((\log W_{\max}) (\log \log W_{\max} +\log n))$ where $W_{\max}$ is the maximal job weight amongst all jobs.

We assume an input sequence with integral release times and weights that are powers of $2$.
This cannot increase the sum of [weighted] completion times by more than a constant factor.

The time slot sequences $R_k$ discussed below are a building block to our menu producing algorithm.

\subsection{Integer sequences $R_i$}

We define integer sequences, $R_i$, $i\geq 0$, defined recursively as follows:
\begin{eqnarray*}
R_0&=&\series{1};\\
\phi_i &=& \series{2^{2^{i-1}}, 2^{2^{i-1}+1},\ldots,2^{2^{i}-1}}, \mbox{\rm\ for $i\geq 1$}; \\
R_i&=&R_{i-1}\|R_{i-1}\|\phi_i, \mbox{\rm\ for $i\geq 1$};
\end{eqnarray*}
Also, let $R_\infty$ be the infinite sequence, such that for every $k$, $R_k$ is a prefix of $R_\infty$. {\sl I.e.},
\begin{eqnarray*}
R_0&=&\series{1}\\
R_1&=& R_0\|R_0\|\series{2^1} = \series{1,1,2}\\
R_2&=& R_1\|R_1\|\series{2^2,2^3} = \series{1,1,2,1,1,2,4,8}\\
&\vdots&\\
R_\infty&=&\series{1,1,2,1,1,2,4,8,1,1,2,1,1,2,4,8,16,32,64,128,1,\ldots}
\end{eqnarray*}

These $R_i$ integer sequences are somewhat analogous to the $S_i$ sequences of Section~\ref{sec:dynamic}.

We now state a few simple properties of these $R_i$ sequences.
\begin{observation}\label{obs:weighted_obs}
	\begin{enumerate}
		\item The length of the sequence $\phi_i$, denoted $\left| \phi_i \right|$, equals $2^{i-1}$.
		\item None of the integers in $\phi_i$ appear in $R_{i-1}$. 
		\item The length of the sequence $R_i$ equals $(i+2)2^{i-1}$. This follows by induction.
	\end{enumerate}	
\end{observation}

We say that $R_j$ appears in $R_i$, $j\leq i$ if there is a contiguous subsequence of $R_i$ equal to $R_j$. Note that, by construction, if $R_j$ appears multiple times in $R_i$, these appearances cannot overlap.

It also follows by the recursive construction that there are $2^i$ appearances of $R_0$ in $R_i$, $2^{i-1}$ appearances of $R_1$ in $R_i$, and in general, $2^{i-k}$ appearances of $R_k$ in $R_i$, for all $1 \leq k \leq i$. Ergo, for $1\leq k \leq i$, every integer $z\in \{2^{2^{k-1}}, 2^{2^{k-1}+1},\ldots, 2^{2^k-1}\}$ appears $2^{i-k}$ times in $R_i$, and for $k=0$, $2^{2^k-1}=1$ appears $2^{i-k}=2^i$ times in $R_i$ .
\begin{corollary}\label{cor:weights_appearances}
	There are exactly $2^{i-1-\floor{\log k}}$ appearances of $2^k$ in $R_i$, for every $i\geq \floor{\log k}+1$.
\end{corollary}

Given an integer sequence $R$, let $R[i]$ denote the $i$th element of $R$, $1\leq i \leq |R|$.


Upon the arrival of a job $j$ at time $t$ with weight $2_j$, our goal is to set prices $\pi_{iq}$ for time slots $[i,i+1]$, $i\geq t$,  on machines $1\leq q \leq m$, so that job $j$ will choose a time slot $[i,i+1]$ and machine $1\leq q \leq m$ if and only if:
\begin{enumerate}
\item Time slot $[i,i+1]$ on machine $q$ is {\sl relevant} upon the arrival of job $j$, i.e., $i\geq t$, and the time slot $[i,i+1]$ on machine $q$ is unoccupied.
\item The weight $w_j$ is ``high enough", {\sl i.e.}, $w_j\geq R_\infty[i]$.
\item There is no earlier time slot $[i',i'+1]$ and machine $q'$ that fulfill the two previous conditions, i.e.,
$t\leq i' < i$, $[i',i'+1]$ on machine $q'$ is unoccupied, and $w_j\geq R_\infty[i']$.
\end{enumerate}
We say time slot $[i,i+1]$ has a {\sl threshold} of $R_\infty[i]$, to refer to the second condition above. Our pricing scheme implies that a job $j$ with weight $w_j$ will choose the earliest unoccupied time slot with a threshold $\leq w_j$.
$X^j$ keeps track of all previously allocated intervals (in all machines): $([i,i+1],q)\in X^{j}$ means that some job $j'<j$ chose the interval $[i,i+1]$ on machine $1 \leq q \leq m$. $X^0=\emptyset$.
Note that there may be several jobs that arrive at time $t$, they appear in some arbitrary order, menu prices are recomputed after every job makes its choice. Menu prices also change over time.

\vspace{0.1in}
\noindent{\bf The Time Slot Pricing Algorithm}
\vspace{0.1in}

Upon the arrival of job $j$ at time $r_j=t$, we compute prices for time slots $[i,i+1]$, $i\geq t$, on machines $1\leq q\leq m$, as follow:
\newline
\MyFrame{
\begin{itemize}
	\item Set $v_1=1$
	\item Let $b_1=\min\{b\geq t|R_\infty[b]=v_1\}$.
	\item For all machines $1\leq q\leq m$, for all $b_1\leq x$, if $[x,x+1]$ is unoccupied on $q$ (i.e, $([x,x+1],q)\notin X^{j-1}$) set the price for time slot $[x,x+1]$ on machine $q$ to be $\pi_1=0$. {\sl I.e.}, effectively add $([x,x+1], q, 0)$ to the menu.
	\item Set $i=1$
	\item Repeat until $b_i-t<1$:
	\begin{itemize}
		\item Let $v_{i+1}=\min R_\infty[t: b_{i}-1]$.
		\item Let $b_{i+1}=\min\{b\geq t|R_\infty[b]=v_{i+1}\}$.
		\item For all machines $1\leq q\leq m$, for all $b_{i+1}\leq x < b_i$, if $[x,x+1]$ is unoccupied on $q$ (i.e, $([x,x+1],q)\notin X^{j-1}$) set the price for time slot $[x,x+1]$ on machine $q$ to be $\pi_{i+1}=\pi_i+(b_{i}-b_{i+1})v_{i+1}$. {\sl I.e.}, this effectively adds $([x,x+1], q, \pi_{i+1})$ to the menu.
		\item Set $i_{\max} = i$,  $i=i+1$.		
	\end{itemize}
\end{itemize}

}\newline

By construction, no job will ever choose a time slot that starts before the job arrival time, nor will it ever choose a slot that has already been chosen. Note that, for all $1\leq q \leq q' \leq m$, if $([i,i+1],q,\pi_{iq})$ and $([i,i+1],q',\pi_{iq'})$ are in the menu, then $\pi_{iq}=\pi_{iq'}$. We refer to the  price of time slot  $[i,i+1]$ as $\pi_i$.  This means that for all $q$ $\pi_{iq}=\pi_i$, excepting, possibly, machines $1\leq q'\leq m$ where slot $[i,i+1]$ is occupied on machine $q'$ and thus does not appear in the menu.


\begin{lemma}\label{lem:firstslot} 
A rational selfish job of weight $w$ always chooses a menu entry of the form $([i,i+1],q,\pi_{i}$) where $R_\infty[i]\leq w$ and $[i,i+1]$ is the earliest relevant time slot in the menu with $R_\infty[i]\leq w$.
\end{lemma}
\begin{proof}
Note that the sequence $\series{v_1,v_2,\ldots }$ is (strongly) monotonically increasing, while  the sequence $\series{b_1,b_2,\ldots}$ is (strongly) monotonically decreasing. Also, by construction, $v_i=R_\infty[b_i]$ for $1 \leq i \leq i_{\max}$.

For any job $j$, let $$\xi_j(x) = (x+1)\cdot w_j+\pi_x$$ be the cost for job $j$ of interval $[x,x+1]$ (if $([x,x+1],q)$ on the menu for some $q$). 
For every $1\leq i \leq i_{\max}-1$ and every $b_{i+1}\leq x < b_{i}$ the cost job $j$ accumulates from time slot $[b_{i+1},b_{i+1}+1]$ is not greater from the cost accumulated by time slot $[x,x+1]$, i.e., $\xi_j(b_{i+1})\leq \xi_j(x)$.
This follows since all such time slots $[x,x+1]$ have the same price, and time slot $[x,x+1]$ is no earlier than time slot $[b_{i+1},b_{i+1}+1]$.
 
 Let $j$ be a job of weight $w_j$, then
\[
\xi_j(b_{i+1})-\xi_j(b_i)=\pi_{i+1}-\left(\pi_{i}+(b_i-b_{i+1})w_j\right) = (b_i-b_{i+1})(v_{i+1}-w_j).
\]
Thus, if $w_j\geq v_{i+1}=R_\infty[b_{i+1}]$, then job $j$ prefers time slot $[b_{i+1},b_{i+1}+1]$ over time slot $[b_i,b_i+1]$. If $w_j < v_{i+1}$ then job $j$ prefers time slot $[b_i,b_i+1]$ over time slot $[b_{i+1},b_{i+1}+1]$.  

In summary, job $j$ prefers the earliest time slot on the menu, $[x,x+1]$, where $w_j\geq R_\infty(x)$.
\end{proof}
%
%
%

\vspace{0.1in}
\noindent{\bf Approximating the minimal sum of weighted completion times.}

\begin{lemma}\label{lem:weights_per_job}
	Let $\sigma=(r_j,w_j,p_j=1)_{j=1}^n$ be an input sequence of jobs. Let $W_{\max}=\max_{1\leq i\leq n}w_i$. Let $j$ be a job in the input sequence with $w_j=2^k$. Denote by $c_j,c^*_j$ the completion time of job $j$ in ALG and WSRPT respectively.
	Then $c_j =O\left(\wmax \left(\log \wmax +\log n \right)\right)\cdot c^*_j.$
\end{lemma}
\begin{proof}
	If $c_j=r_j+p_j=r_j+1$ the claim is obviously true.
	Otherwise, $c_j\geq r_j+2$, and there exists integers $a,d\geq 0$ such that $\len{R_a}\leq r_j+p_j=r_j+1<\len{R_{a+1}}$ and $\len{R_{a+d}}< c_j\leq\len{R_{a+d+1}}$. By Observation~\ref{obs:weighted_obs}, $c_j\leq (a+d+3)2^{a+d}$, and $c^*_j\geq (a+2)2^{a-1}$.
	
	We consider the following cases:
	\begin{enumerate}[a.]
	\item If $R_{a+d+1}$ is the sequence of minimal index that includes an entry $2^k$, then by Corollary~\ref{cor:weights_appearances} $a+d+1=\floor{\log k}+1$, and $W_{\max}\geq 2^k$, thus
	\[c_j\leq (\floor{\log k}+3)2^{\floor{\log k}}=O\left(\wmax\log \wmax\right),\] as $c^*_j\geq 1$ the claim holds.
	
	\item If $d\leq 1$ then
	\[c_j\leq (a+4)2^{a+1}\leq 8c^*_j.\]
	
	\item If neither (a) nor (b) are true, it follows that $a+d+1\geq\floor{\log k}+2$ and that $d\geq 2$. Recall that the prefix of $R_{a+d}$ is a concatenation of two $R_{a+d-1}$ sequences. As $d\geq 2$, $a+d-1\geq a+1$. As $a+d+1\geq\floor{\log k}+2$, $a+d\geq \floor{\log k}+1$ and by Corollary~\ref{cor:weights_appearances}, there are at least $2^{a+d-\floor{\log k}}-2^{a+1-\floor{\log k}}\geq 2^{a+d-1-\floor{\log k}}$ appearances of time slots with $2^k$ threshold in $\left[\len{R_{a+1}},\len{R_{a+d}}\right]$ on each machine. All these time slots are occupied on all of the machines, otherwise job $j$ would prefer such an unoccupied time slot over his choice. 
	
	Thus, there are at least $m2^{a+d-1-\floor{\log k}}$ jobs of weight $\geq 2^k$ that arrived before job $j$, i.e., jobs that are completed no later than job $j$ in WSRPT, and $c^*_j\geq 2^{a+d-1-\floor{\log k}}$. Ergo,
	\[c_j\leq (a+d+3)2^{\floor{\log k}+1}c^*_j.\]
	If $a+d \geq \floor{\log k}+1$ then $n\geq m2^{a+d-1-\floor{\log k}}\geq m2^{\nicefrac{1}{2}(a+d)}$ and $(a+d+3)=O(\log n)$.
	If on the other hand, $a+d < \floor{\log k}+1$ then $(a+d+3)=O(\log \wmax)$. In both these cases we have that
	\[c_j=O\left(\wmax\left(\log \wmax+\log n\right)\right)\] and the claim holds.
	\end{enumerate}

\end{proof}
\begin{theorem}
	$\mbox{\rm Cost}(\mbox{\rm ALG})=O\left(\wmax \left(\log \wmax +\log n \right)\right)\cdot \mbox{\rm Cost}(\mbox{\rm OPT})$
\end{theorem}
\begin{proof}
	The theorem is achieved by applying Lemma~\ref{lem:weights_per_job} on every job individually, and by the fact that WSRPT is a 2 approximation to the optimal preemptive offline algorithm.
\end{proof}

%
%
%
%

\section{Lower Bound on the Competitive Ratio for any Prompt Online Algorithm, Arbitrary Weights}
\begin{theorem}
	Any  prompt online algorithm is $\Omega\left(\log W_{\max}\right)$  competitive, even if randomization is allowed.
\end{theorem}

\begin{proof}
Consider the following sequence for some large $k$:\newline
\MyFrame{
	For $j=1,\ldots, 8k$:
	\begin{itemize}
		\item Job $j$, of weight $w_j=2^{k+j}$, arrives at time zero.
		
		\item If $\expect{}{c_j} > 4k$, let $j^*=j$. Stop generating new jobs.
	\end{itemize}
}\newline

Note that $W_{\max}\leq 2^{9k}$ in this sequence.

\begin{lemma}
	There must be an iteration $j\in \{1,\ldots ,8k \}$ for which $\expect{}{c_j} > 4k$. \label{lem:overflow_weights}
\end{lemma}
\begin{proof}
	Assume that for every $j\in \{1,\ldots ,8k \}$ $\expect{}{c_j}\leq 4k$. Then  $\expect{}{\sum_{i=1}^{8k}c_i}\leq 4k\cdot 8k$. Therefore, there is some realization of the algorithm in which the sum of completion times is at most $4k\cdot 8k$. However, for any valid schedule, the sum of completion times for $8k$ unit size jobs is at least $\sum_{i=1}^{8k}i=4k(8k+1)$, a contradiction.
\end{proof}

	According to the Lemma above, the last job in the sequence, job $j^*$ has $\expect{}{c_{j^*}} > 4k$.  For ALG, $\expect{}{c_{j^*}} > 4k$ implies
	$$\expect{}{\mbox{\rm Cost}(\mbox{\rm ALG})}\geq\expect{}{w_{j^*}c_{j^*}}> 2^{k+j^*+2}k.\label{eq:alg_lb}$$
	In OPT, the jobs are scheduled from the largest weighted job (of weight $w_{j^*}=2^{k+j^*}$) to the smallest weighted job (of size $w_1=2^{k+1}$). Therefore,	
	\[\mbox{\rm Cost}(\mbox{\rm OPT}) = \sum_{i=0}^{j^*-1}(i+1)2^{k+j^*-i} = 2^{k+j^*}\sum_{i=0}^{j^*-1}\frac{i+1}{2^i}< 2^{k+j^*+2}.\]
	
	Thus, $\expect{}{\mbox{\rm Cost}(\mbox{\rm ALG})/\mbox{\rm Cost}(\mbox{\rm OPT})}> k = \Omega{(\log W_{\max})}$.
\end{proof}


\bibliographystyle{plainnat}

\bibliography{prompt_scheduling}

\begin{thebibliography}{22}
\providecommand{\natexlab}[1]{#1}
\providecommand{\url}[1]{\texttt{#1}}
\expandafter\ifx\csname urlstyle\endcsname\relax
  \providecommand{\doi}[1]{doi: #1}\else
  \providecommand{\doi}{doi: \begingroup \urlstyle{rm}\Url}\fi

\bibitem[Bruno et~al.(1974)Bruno, Coffman, and Sethi]{Bruno74}
J.~Bruno, E.~G. Coffman, Jr., and R.~Sethi.
\newblock Scheduling independent tasks to reduce mean finishing time.
\newblock \emph{Commun. ACM}, 17\penalty0 (7):\penalty0 382--387, July 1974.
\newblock ISSN 0001-0782.
\newblock \doi{10.1145/361011.361064}.

\bibitem[Christodoulou et~al.(2004)Christodoulou, Koutsoupias, and
  Nanavati]{CKN04}
George Christodoulou, Elias Koutsoupias, and Akash Nanavati.
\newblock Coordination mechanisms.
\newblock In \emph{Automata, Languages and Programming: 31st International
  Colloquium, {ICALP} 2004, Turku, Finland, July 12-16, 2004. Proceedings},
  pages 345--357, 2004.
\newblock \doi{10.1007/978-3-540-27836-8_31}.

\bibitem[Cohen et~al.(2015)Cohen, Eden, Fiat, and Jez]{CohenEFJ15}
Ilan~Reuven Cohen, Alon Eden, Amos Fiat, and Lukasz Jez.
\newblock Pricing online decisions: Beyond auctions.
\newblock In \emph{Proceedings of the Twenty-Sixth Annual {ACM-SIAM} Symposium
  on Discrete Algorithms, {SODA} 2015, San Diego, CA, USA, January 4-6, 2015},
  pages 73--91, 2015.
\newblock \doi{10.1137/1.9781611973730.7}.

\bibitem[Daskalakis et~al.(2017)Daskalakis, Babaioff, and
  Moulin]{DBLP:conf/sigecom/2017}
Constantinos Daskalakis, Moshe Babaioff, and Herv{\'{e}} Moulin, editors.
\newblock \emph{Proceedings of the 2017 {ACM} Conference on Economics and
  Computation, {EC} '17, Cambridge, MA, USA, June 26-30, 2017}, 2017. {ACM}.
\newblock ISBN 978-1-4503-4527-9.
\newblock \doi{10.1145/3033274}.

\bibitem[Feldman et~al.(2017)Feldman, Fiat, and Roytman]{FFR17}
Michal Feldman, Amos Fiat, and Alan Roytman.
\newblock Makespan minimization via posted prices.
\newblock In  \citet{DBLP:conf/sigecom/2017}, pages 405--422.
\newblock ISBN 978-1-4503-4527-9.
\newblock \doi{10.1145/3033274.3085129}.

\bibitem[Gkatzelis et~al.(2017)Gkatzelis, Markakis, and
  Roughgarden]{GkatzelisMR17}
Vasilis Gkatzelis, Evangelos Markakis, and Tim Roughgarden.
\newblock Deferred-acceptance auctions for multiple levels of service.
\newblock In  \citet{DBLP:conf/sigecom/2017}, pages 21--38.
\newblock ISBN 978-1-4503-4527-9.
\newblock \doi{10.1145/3033274.3085142}.

\bibitem[Graham(1966)]{G66}
R.~L. Graham.
\newblock {Bounds for certain multiprocessing anomalies}.
\newblock \emph{Bell System Technical Journal}, 45:\penalty0 1563--1581, 1966.

\bibitem[Graham et~al.(1979)Graham, Lawler, Lenstra, and Kan]{Graham79}
Ronald~L Graham, Eugene~L Lawler, Jan~Karel Lenstra, and AHG~Rinnooy Kan.
\newblock Optimization and approximation in deterministic sequencing and
  scheduling: a survey.
\newblock In \emph{Annals of discrete mathematics}, volume~5, pages 287--326.
  Elsevier, 1979.

\bibitem[Hall et~al.(1997)Hall, Schulz, Shmoys, and Wein]{Hall97}
Leslie~A. Hall, Andreas~S. Schulz, David~B. Shmoys, and Joel Wein.
\newblock Scheduling to minimize average completion time: Off-line and on-line
  approximation algorithms.
\newblock \emph{Math. Oper. Res.}, 22\penalty0 (3):\penalty0 513--544, 1997.
\newblock \doi{10.1287/moor.22.3.513}.

\bibitem[Hartline and Roughgarden(2009)]{HR09}
Jason~D. Hartline and Tim Roughgarden.
\newblock Simple versus optimal mechanisms.
\newblock In \emph{Proceedings 10th {ACM} Conference on Electronic Commerce
  (EC-2009), Stanford, California, USA, July 6--10, 2009}, pages 225--234,
  2009.
\newblock \doi{10.1145/1566374.1566407}.

\bibitem[Im and Kulkarni(2016)]{ImK16}
Sungjin Im and Janardhan Kulkarni.
\newblock Fair online scheduling for selfish jobs on heterogeneous machines.
\newblock In \emph{Proceedings of the 28th {ACM} Symposium on Parallelism in
  Algorithms and Architectures, {SPAA} 2016, Asilomar State Beach/Pacific
  Grove, CA, USA, July 11-13, 2016}, pages 185--194, 2016.
\newblock \doi{10.1145/2935764.2935773}.

\bibitem[Im et~al.(2017)Im, Moseley, Pruhs, and Stein]{IMPS17}
Sungjin Im, Benjamin Moseley, Kirk Pruhs, and Clifford Stein.
\newblock Minimizing maximum flow time on related machines via dynamic posted
  pricing.
\newblock In Kirk Pruhs and Christian Sohler, editors, \emph{25th Annual
  European Symposium on Algorithms, {ESA} 2017, September 4-6, 2017, Vienna,
  Austria}, volume~87 of \emph{LIPIcs}, pages 51:1--51:10. Schloss Dagstuhl -
  Leibniz-Zentrum fuer Informatik, 2017.
\newblock ISBN 978-3-95977-049-1.
\newblock \doi{10.4230/LIPIcs.ESA.2017.51}.

\bibitem[Immorlica et~al.(2009)Immorlica, Li, Mirrokni, and Schulz]{ILMS09}
Nicole Immorlica, Li~(Erran) Li, Vahab~S. Mirrokni, and Andreas~S. Schulz.
\newblock Coordination mechanisms for selfish scheduling.
\newblock \emph{Theor. Comput. Sci.}, 410\penalty0 (17):\penalty0 1589--1598,
  2009.
\newblock \doi{10.1016/j.tcs.2008.12.032}.

\bibitem[Lenstra et~al.(1990)Lenstra, Shmoys, and Tardos]{Lenstra90}
Jan~Karel Lenstra, David~B. Shmoys, and {\'{E}}va Tardos.
\newblock Approximation algorithms for scheduling unrelated parallel machines.
\newblock \emph{Math. Program.}, 46:\penalty0 259--271, 1990.
\newblock \doi{10.1007/BF01585745}.

\bibitem[Lenstra et~al.(1977)Lenstra, Kan, and Brucker]{Lenstra1977}
J.K. Lenstra, A.H.G.~Rinnooy Kan, and P.~Brucker.
\newblock Complexity of machine scheduling problems.
\newblock In P.L. Hammer, E.L. Johnson, B.H. Korte, and G.L. Nemhauser,
  editors, \emph{Studies in Integer Programming}, volume~1 of \emph{Annals of
  Discrete Mathematics}, pages 343 -- 362. Elsevier, 1977.
\newblock \doi{https://doi.org/10.1016/S0167-5060(08)70743-X}.

\bibitem[Megow and Schulz(2004)]{Megow04}
Nicole Megow and Andreas~S. Schulz.
\newblock On-line scheduling to minimize average completion time revisited.
\newblock \emph{Oper. Res. Lett.}, 32\penalty0 (5):\penalty0 485--490, 2004.
\newblock \doi{10.1016/j.orl.2003.11.008}.

\bibitem[Nisan and Ronen(2001)]{NR01}
Noam Nisan and Amir Ronen.
\newblock Algorithmic mechanism design.
\newblock \emph{Games and Economic Behavior}, 35\penalty0 (1-2):\penalty0
  166--196, 2001.

\bibitem[Phillips et~al.(1998)Phillips, Stein, and Wein]{Phillips1998}
Cynthia Phillips, Clifford Stein, and Joel Wein.
\newblock Minimizing average completion time in the presence of release dates.
\newblock \emph{Mathematical Programming}, 82\penalty0 (1):\penalty0 199--223,
  Jun 1998.
\newblock ISSN 1436-4646.
\newblock \doi{10.1007/BF01585872}.

\bibitem[Schrage(1968)]{Schrage68}
Linus Schrage.
\newblock Letter to the editor-a proof of the optimality of the shortest
  remaining processing time discipline.
\newblock \emph{Operations Research}, 16\penalty0 (3):\penalty0 687--690, 1968.

\bibitem[Schrage and Miller(1966)]{Schrage66}
Linus~E. Schrage and Louis~W. Miller.
\newblock The queue m / g /1 with the shortest remaining processing time
  discipline.
\newblock \emph{Operations Research}, 14\penalty0 (4):\penalty0 670--684, 1966.

\bibitem[Shmoys et~al.(1995)Shmoys, Wein, and Williamson]{Shmoys95}
David~B. Shmoys, Joel Wein, and David~P. Williamson.
\newblock Scheduling parallel machines on-line.
\newblock \emph{{SIAM} J. Comput.}, 24\penalty0 (6):\penalty0 1313--1331, 1995.
\newblock \doi{10.1137/S0097539793248317}.

\bibitem[Smith(1956)]{Smith56}
Wayne~E. Smith.
\newblock Various optimizers for single-stage production.
\newblock \emph{Naval Research Logistics Quarterly}, 3\penalty0 (1-2):\penalty0
  59--66, 1956.
\newblock ISSN 1931-9193.
\newblock \doi{10.1002/nav.3800030106}.

\end{thebibliography}
\newpage
\appendix

\section{The Static Mechanism Upper Bound}
In this section, we analyze a mechanism based on a single component of Section~\ref{sec:dynamic}.
We divide the timeline interval $[0,\infty]$ into time intervals as in $S_\infty(0)$, in every machine. This division does not change over time, i.e., this is a static mechanism. Upon arrival, a job may choose an unoccupied interval in $S_\infty(0)$ in some machine. Denote this mechanism as {\sl ALG}.

Let $\sigma=\left(r_j,p_j\right)_{j=1}^n$ be a sequence of jobs.
Let $P_{\max}={\max}_{j\in[n]}p_j$ be the maximal processing time among all jobs. Recall that we may assume that all jobs processing time are of length $2^k$ for some $k$. Therefore, there are at most $\log P_{\max}+1$ different processing times for jobs in $\sigma$. Let $n_i=\left|\left\{j|2^{i}<p_j\leq2^{i+1}\right\}\right|$ for $0\leq i \leq \log P_{\max}$. Let $n_{\max}={\max}_{0\leq i \leq \log P_{\max}}n_i$.
We prove our static mechanism achieves the following competitive ratio:

\begin{theorem}
	The competitive ratio of ALG is $O(\log P_{\max}+\log n_{\max})$\label{thm:static_cr}
\end{theorem}
\begin{proof}
	
	Let $c^*_j$ be the completion time of job $j$ in SRPT.
	We prove that for every job $j$, $$c_j\leq O(\log P_{\max}+\log n_{\max})\cdot c^*_j.$$
	Let $j$ be a job in the input sequence with $p_j=2^k$.

	\begin{itemize}
		
		\item If $j$ is assigned in $S_k(0)$ for some machine (i.e., in the first $2^k$ slot in that machine), then by Corollary~\ref{cor:len_sk} $c_j=(k+1)2^k$, while $c^*_j\geq2^k$. $P_{\max}\geq p_j= 2^k$. Therefore, $$c_j \leq \frac{(k+1)2^k}{2^k}c^*_j = (k+1)c^*_j = O\left(\log P_{\max}\right)c^*_j.$$
		\item Otherwise, as $p_j\geq 1$, there exists some $d\geq 0$ with $e(S_d(0))\leq r_j+p_j<e(S_{d+1}(0))$. Let $\ell\geq 0$ be such that $e(S_{d+\ell}(0))< c_j \leq e(S_{d+\ell+1}(0))$. Notice that $k\leq d+\ell$ as $j$ did not choose the first $2^k$ interval. We look at two possible cases:	
		\begin{itemize}
			\item {\boldmath $\ell \leq 1$}. In this case by Corollary~\ref{cor:len_sk} $c_j\leq (d+3)2^{d+2}$ while $c^*_j\geq(d+1)2^d$. Thus, $$c_j\leq 4\cdot \frac{d+3}{d+1}\cdot c^*_j \leq 12c^*_j.$$
			\item {\boldmath $\ell \geq 2$}. Then for every machine, all intervals of size $2^k$ that are in $S_{d+\ell}(0)\setminus S_{d}(0)$ are occupied. By Corollary~\ref{cor:len_sk} there are $2^{d+\ell-k}$ different appearances of $S_k$ in $S_{d+\ell}(0)$. There are at most $2^{d+1-k}$ appearances of $S_k$ in $S_{d+1}(0)$ (only if $d+1\geq k$). Using Lemma~\ref{lem:caught_by_smaller} on every such $S_k$ appearance on every machine separately, suggests that 
			$$vol\left(D(j)\cap \left\{j'|I(j')\in S_{d+\ell}(0)\setminus S_{d+1}(0), M(j')=q \right\}\right) \geq 2^{d+\ell}-2^{d+1}\geq 2^{d+\ell-1},$$ (as $\ell \geq 2$), i.e.,
			the total volume of jobs no greater than $2^k$ on every machine in $\left[e(S_{d+1}(0))_,e(S_{d+\ell}(0))\right]$ that arrived no later then job $j$, is at least $2^{d+\ell-1}$. Ergo, there is an overall total volume of $m2^{d+\ell-1}$. 
			
			As there are $k+1$ different processing times possible for jobs with processing time $\leq 2^k$, we get that $P_{\max}\cdot n_{\max}\geq \frac{m2^{d+\ell-1}}{k+1}\geq \frac{m2^{d+\ell-1}}{d+\ell+1}\geq m\cdot 2^{\sfrac{d+\ell}{2}-1}$.
			In SRPT, job $j$ will be scheduled after a volume of at least $m2^{d+\ell-1}$ scattered among $m$ machines, therefore $c^*_j\geq 2^{d+\ell-1}$. Thus,
			\begin{eqnarray*}
				c_j&\leq& (d+\ell+2)\frac{2^{d+\ell+1}}{2^{d+\ell-1}}c^*_j \\
				&\leq& 4(d+\ell+2)c^*_j \leq 8(d+\ell)c^*_j \\
				&=&	O\left(\log P_{\max}+\log n_{\max}\right)c^*_j.
			\end{eqnarray*}
			
		\end{itemize}		
	\end{itemize}
\end{proof}

\section{A Lower Bound for the Static Mechanism}\label{sec:static_lb}
In this section we show that, as a function of $P_{\max}$ (without $n_{\max}$), and with jobs feedback (i.e., a job does not occupy an entire interval if it is larger than the job size) the competitive ratio of the static mechanism is $\Omega\left(\sqrt{P_{\max}}\right)$.
Consider the following input for some fixed $n,k$ with $k<n$:\newline
\MyFrame{
\begin{itemize}
	\item First, $(n-k)2^{n-k}$ jobs with processing time $2^k$ arrive one after the other (at time 0).
	\item Set $i=k-1$.
	\item While $i\geq 0$:
	\begin{itemize}
		\item $2^{n-i}$ jobs with processing time $2^i$ arrive one after the other (at time 0).
		\item Set $i=i-1$
	\end{itemize}
	\item Finally, $2^n$ jobs with processing time 1 arrive one after the other (at time 0)
\end{itemize}
}\newline

The input described may be described as follows - at first we fill all intervals in $S_n(0)$ with jobs of size at most $2^k$. We fill every interval of length $\geq 2^k$ with jobs of size $2^k$ ($2^i$ jobs for interval of length $2^{k+i}$). Recall that the sum of the length of intervals in $S_n(0)$ equals $(n+1)2^n$ (Corollary~\ref{cor:len_sk}), and that for every $0\leq i\leq n$ the total length of intervals of length $2^i$ in $S_n(0)$ is $2^n$. It follows that the total length of intervals of lengths $\geq 2^k$ in $S_n(0)$ is $(n-k)2^n$. We then fill every $2^i$ interval for $0\leq i\leq k-1$ with a job of processing time $2^i$. After all intervals in $S_n(0)$ are occupied, we add $2^n$ jobs with processing time 1.

The last $2^n$ jobs are all completed (and start) later than $e(S_n(0))=(n+1)2^n$. Thus, $Cost(ALG)> n2^{2n}$.
The optimal algorithm, SPT schedules the jobs by increasing order of their sizes. We denote the cost induced by all jobs of size $2^i$ by $C_i$. The cost of a single job of size $2^i$ is the total processing time of all smaller jobs plus the total processing time of $2^i$ sized jobs that arrived before it, i.e, the completion time of the $j$th job of size $2^i$ is $i2^n+\sum_{k=1}^{j}j2^i$. This implies the following cost function: 
\begin{eqnarray*}
	Cost(OPT)&=&C_0+\sum_{i=1}^{k-1}C_i+C_k\\
	&=&\sum_{j=1}^{2^{n+1}}j+\sum_{i=1}^{k-1}\left(\sum_{j=1}^{2^{n-i}}\left(i2^n+j2^i\right)\right) + \sum_{j=1}^{(n-k)2^{n-k}}\left(k2^n+j2^k\right)\\
	&\leq&  2^{2n+2}+\sum_{i=1}^{k-1}\left(2^{2n}\frac{i}{2^i}+2^{2n-i}\right)+k(n-k)2^{2n-k}+(n-k)^22^{2n-k}\\
	&\leq& 2^{2n+2}+2^{2n}\sum_{i=1}^{k-1}\frac{i+1}{2^i}+n(n-k)2^{2n-k}\\
	&\leq& 2^{2n}\left(7+\frac{n(n-k)}{2^k}\right)
\end{eqnarray*}
For $k=2\log n$ we get that $Cost(OPT)\leq 8\cdot2^{2n}$ while $Cost(ALG)\geq n2^{2n}$ and $P_{\max}=n^2$. Thus, the competitive ratio is $\Omega\left(\sqrt{P_{\max}}\right)$.

\end{document}